%% file: main.tex
\documentclass[a4paper,10pt]{article}

\usepackage[margin=90pt]{geometry}

\usepackage{amsmath}
\usepackage{amssymb}
\usepackage{amsthm}
\usepackage{url}
\usepackage[all]{xy}
\usepackage{dsfont}
\usepackage[linesnumbered,lined,ruled]{algorithm2e}
\usepackage{wrapfig}
\usepackage{faktor}
\usepackage{makecell}
\usepackage[inline]{enumitem}
\usepackage{graphicx}
\usepackage{adjustbox}

\usepackage{color}

\newcommand{\quiv}{\mathcal{Q}}

\newcommand{\Vmod}{\mathbb{V}}
\newcommand{\Wmod}{\mathbb{W}}
\newcommand{\Imod}{\mathbb{I}}

\newcommand{\im}{\mathrm{im\ }}

\newcommand{\Z}{\mathbb{Z}}					
\newcommand{\R}{\mathbb{R}}					
\newcommand{\Rips}{\mathcal{R}}
\newcommand{\PID}{\mathbf{D}}

\newcommand{\A}{\mathcal{A}}				
\newcommand{\Q}{\mathcal{Q}}				
\newcommand{\K}{\mathcal{K}}				
\newcommand{\Ah}{\hat{\mathcal{A}}}			
\newcommand{\Qh}{\hat{\mathcal{Q}}}			
\newcommand{\Kh}{\hat{\mathcal{K}}}			
\newcommand{\wh}{\hat{\w}}

\newcommand{\hasse}{\mathcal{H}}

\newcommand{\morsepart}{\A \sqcup \Q \sqcup \K}

\newcommand{\morse}[1]{
	(\A_{#1},\allowbreak \Q_{#1},\allowbreak \K_{#1}, \w_{#1})}	
\newcommand{\morseh}[1]{
	(\Ah_{#1},\allowbreak \Qh_{#1},\allowbreak \Kh_{#1},\allowbreak \wh_{#1})}	
\newcommand{\w}{\omega}

\newcommand{\clx}{X}				
\newcommand{\oclx}{\overline{\clx}}
\newcommand{\field}{\mathbb{F}}				
\newcommand{\F}{\mathcal{F}}				
\newcommand{\oF}{\overline{\F}}				
\newcommand{\M}{\mathcal{M}}				
\newcommand{\bM}{\overline{\M}}
\newcommand{\homfun}[1]{H(#1, \field)}		
\newcommand{\hf}[1]{H(#1)}					
\newcommand{\Hom}{H}
\newcommand{\homfund}[2]{H_{#2}(#1, \field)}	
\newcommand{\id}{\mathds{1}}				
\newcommand{\bo}{\partial}					
\newcommand{\mbo}[1]{\partial^{#1}}		

\newcommand{\isp}[2]{\left\langle #1, #2 \right\rangle}	
\newcommand{\inc}[2]{\left[ #1 : #2 \right]}
\newcommand{\mtx}{\mathcal{B}}				
\newcommand{\bmtx}{\overline{\mtx}}

\newcommand{\cycle}{Z}
\newcommand{\border}{B}
\newcommand{\chain}{C}

\newcommand{\orac}{\mathcal{C}}

\theoremstyle{plain}
\newtheorem{theorem}{Theorem}
\newtheorem{lemma}[theorem]{Lemma}
\newtheorem{prop}{Properties}
\newtheorem{notation}{Notations}
\newtheorem{definition}{Definition}

\theoremstyle{remark}
\newtheorem{remark}{Remark}

\setlist[enumerate,1]{label={\emph{(\arabic*)}}}
\setlist[enumerate]{wide, labelwidth=!, labelindent=0pt, itemsep=0px}

\newcommand\restr[2]{{
  \left.\kern-\nulldelimiterspace 
  #1 
  \vphantom{\big|} 
  \right|_{#2} 
  }}
  
\title{Discrete Morse Theory for Computing Zigzag Persistence}
\author{Cl\'ement Maria 
	\thanks{INRIA Sophia Antipolis-M\'editerran\'ee, France 
	-- \texttt{clement.maria@inria.fr}} 
	\and Hannah Schreiber 
	\thanks{Graz University of Technology, Austria 
	-- \texttt{hschreiber@tugraz.at}
	-- Supported by the Austrian Science Fund (FWF) grant number P 29984-N35.}
}
\date{\today}

\begin{document}

\maketitle


\begin{abstract}
	We introduce a theoretical and computational framework to use discrete 
	Morse theory as an efficient preprocessing in order to compute zigzag 
	persistent homology. 
	From a zigzag filtration of complexes $(\clx_i)$, we introduce a 
	\emph{zigzag Morse filtration} whose complexes $(\A_i)$ are Morse reductions of 
	the original complexes $(\clx_i)$, and we prove that they both have same 
	persistent homology. This zigzag Morse filtration generalizes the 
	\emph{filtered Morse complex} of Mischaikow and Nanda~\cite{MischaikowN13}, 
	defined for standard persistence.

	The maps in the zigzag Morse filtration are forward and backward 
	inclusions, as is standard in zigzag persistence, as well as a new type of 
	map inducing non trivial changes in the boundary operator of the Morse 
	complex. We study in details this last map, and design algorithms to 
	compute the update both at the complex level and at the homology matrix 
	level when computing zigzag persistence. The key point of our construction is 
	that it does not require any knowledge of past and future maps of the input 
	filtration. We deduce an algorithm to compute the zigzag persistence of a 
	filtration that depends mostly on the number of critical cells of the complexes, 
	and show experimentally that it performs better in practice.
\end{abstract}


\section{Introduction}
\label{sec:intro}
\input{introduction}

\section{Background}
\label{sec:background}
\input{background}

\section{Zigzag Morse filtration and persistence}
\label{sec:hl_algo}
\input{hl_algo}

\section{Boundary of the Morse complex}
\label{sec:boundary}
\input{boundary}

\section{Persistence algorithm for zigzag Morse complexes}
\label{sec:ll_algo}
\input{ll_algo}

\section{Experiments}
\label{sec:expe}
\input{experiments}


\bibliographystyle{plain}
\bibliography{bibliography}

\end{document}

%% file: introduction.tex
\emph{Persistent homology} is an algebraic method that permits to characterize 
the evolution of the topology of a growing sequences of spaces 
$\clx_1 \subseteq \ldots \subseteq \clx_n$, called a filtration. 
The theory has found many applications, especially in data analysis where it has 
been successfully applied to material science~\cite{Lee17nanoporous}, shape 
classification~\cite{czcg-pbs-05,ccgmo09}, or 
clustering~\cite{CBPC2013,cgos-pbc-13}.

Filtrations can be represented with help of diagrams as follows:
\begin{equation}\label{eq:ph_filtration}
	\xymatrix @C-5pt{
		\clx_1 \ar@{->}^-\subseteq[r] 
		& \clx_2 \ar@{->}^-\subseteq[r] 
		& \cdots \ar@{->}^-\subseteq[r] 
		& \clx_{n-1} \ar@{->}^-\subseteq[r] 
		& \clx_n
	}.
\end{equation}
Applying a homology functor, for a coefficient field $\field$, to a filtration 
leads to a sequence of vector spaces --- the \emph{homology groups} 
$\homfun{\clx_i}$ --- connected by maps induced by the inclusions, known as a 
\emph{persistence module}:
\begin{equation}\label{eq:persistence_module}
	\xymatrix @C-5pt{
		\homfun{\clx_1} \ar@{->}[r] 
		& \homfun{\clx_2} \ar@{->}[r] 
		& \cdots \ar@{->}[r] 
		& \homfun{\clx_{n-1}} \ar@{->}[r] 
		& \homfun{\clx_n}
	}.
\end{equation}
Computing the persistent homology of a filtration~(\ref{eq:ph_filtration}) 
consists of computing the isomorphism type, known as the 
\emph{interval decomposition}, of its corresponding persistence 
module~(\ref{eq:persistence_module}).

The success of persistent homology relies on sound theoretical 
foundations~\cite{0025666,EdelsbrunnerLZ02,ZomorodianC05}, favorable stability 
properties~\cite{BauerL15,0039900,Cohen-SteinerEH07}, and fast algorithms, 
both theoretically~\cite{ChenK13,Cohen-SteinerEM06,SilvaMV11,MilosavljevicMS11} and 
experimentally~\cite{Bauer:arXiv1303.0477,BauerKR14,BMDcamalgorithmica,Chen11PH}, 
to compute the interval decomposition of an input filtration. This last effort 
towards better implementations has led to dramatic improvements of running 
times in practice, and the emergence of efficient software libraries in the 
field, such as \texttt{Dionysus}~\cite{dionysus_morozov}, 
\texttt{DIPHA}~\cite{dipha_lib}, \texttt{GUDHI}~\cite{gudhi_ICMS14}, 
and \texttt{Ripser}~\cite{ripser_lib}.

Another approach to fast computation consists of preprocessing the input 
filtration~(\ref{eq:ph_filtration}) in order to drastically reduce the size of the 
domains $\clx_i$, while preserving the interval decomposition of the persistence 
module~(\ref{eq:persistence_module})~\cite{BoissonnatPP18,abs-1210-1429,%
MischaikowN13,5766002}. This approach has the double advantage of 
reducing both time and memory complexity. 
This goal has successfully been reached by the use of 
\emph{discrete Morse theory}~\cite{abs-1210-1429,Forman98DMT,MischaikowN13} (see 
also~\cite{CurryGN16,HarkerMMN14}), and led to the implementation of the efficient 
software, such as \texttt{Perseus}~\cite{perseus_cite} and 
\texttt{Diamorse}~\cite{diamorse_cite}. Additionally, noticeable successes, at the 
crossroad of persistence and discrete Morse theory, have been reached in the study of 
3D images~\cite{5766002}, allowing drastic improvements in memory and time 
performance, as well as the study of data ranging from medical imaging to material 
science~\cite{7025987,6873268,6134731}.

\emph{Zigzag persistent homology} is a generalization of persistent homology 
that allows the measurement and tracking of the topology of sequences of spaces 
that both grow and shrink, known as a \emph{zigzag filtrations}:
\begin{equation}\label{eq:zz_filtration}
	\xymatrix @C-5pt{
		\clx_1 \ar@{->}^-\subseteq[r] 
		& \clx_2 \ar@{<-}^-\supseteq[r] 
		& \cdots \ar@{->}^-\subseteq[r] 
		& \clx_{n-1} \ar@{<-}^-\supseteq[r] 
		& \clx_n
	},
\end{equation}
which gives a \emph{zigzag module}, also admitting an interval decomposition:
\begin{equation}\label{eq:zz_module}
	\xymatrix @C-5pt{
		\homfun{\clx_1} \ar@{->}[r] 
		& \homfun{\clx_2} \ar@{<-}[r] 
		& \cdots  \ar@{->}[r] 
		& \homfun{\clx_{n-1}} \ar@{<-}[r] 
		& \homfun{\clx_n}
	}.
\end{equation}
 
The theory of zigzag persistence was introduced in~\cite{CarlssonS10}, and 
theoretical~\cite{MilosavljevicMS11} and 
practical~\cite{CarlssonSM09,MariaO15} algorithms have been introduced to 
compute it. Zigzag persistence has great applicative potential, considering it 
provably produces better topological information in topology 
inference~\cite{os-zz-14}, while maintaining the homology of smaller spaces 
$\clx_i$ thanks to deletions of faces, and more generally allows a finer 
approach to data analysis, such as density estimation and topological 
bootstrapping~\cite{CarlssonS10}.

However, computing zigzag persistence is more intricate that computing 
persistent homology, essentially due to the fact that the full sequence of 
insertions and deletions of faces is unknown, which requires the 
maintenance and update of heavier data structures. As a consequence, none of 
the optimizations of persistence algorithms adapt to the 
zigzag case. The relatively poor performance of zigzag persistence 
implementations, compared with persistent homology ones, is a major hindrance 
to its practical use.

\paragraph{Motivation and applications for zigzag persistence.} 
We give two important applications of zigzag persistence on which we test the 
experimental performance of our method.

\begin{enumerate}[itemsep=\smallskipamount, topsep=\smallskipamount]
	\item \emph{Topology inference from data points $P$.} A standard 
	approach~\cite{0025666} consists of computing the persistent homology of the 
	Rips complex $\Rips^\rho(P)$ on the set of points $P$, for an increasing 
	threshold $\rho \geq 0$. We compute instead the zigzag persistence of 
	oscillating Rips zigzag filtrations~\cite{os-zz-14}. These filtrations add 
	data points progressively while reducing the scale of reconstruction in order 
	to adapt to a more and more dense set of points. Specifically, 
	\begin{equation}\label{eq:zz_filtration_oRzz}
		\xymatrix @C-5pt{
		& \ar@{->}[l]|-{\cdots} \Rips^{\mu \varepsilon_i} (P_i) 
			\ar@{->}^-\subseteq[r] 
		& \Rips^{\nu \varepsilon_i} (P_i \cup \{p_{i+1}\}) \ar@{<-}^-\supseteq[r]	
		& \Rips^{\mu \varepsilon_i} (P_i \cup \{p_{i+1}\}) \ar@{->}[r]|-{\cdots} 
		& \\
		},
	\end{equation}
	where $\Rips^\alpha(P)$ is the Rips complex of threshold $\alpha$ on points 
	$P$, and $\varepsilon_i$ a measure of the ``sparsity'' of the set of points 
	$P_i := \{p_1, \ldots , p_{i}\}$ that decreases when points are added. 
	Finally, $0 < \mu \leq \nu$ are parameters. This filtration is known to 
	furnish provably correct persistence diagrams, with much less noise than 
	standard persistence~\cite{os-zz-14}, while naturally maintaining much 
	smaller complexes during computation. This application is of importance in 
	data analysis~\cite{CBPC2013,cgos-pbc-13}.
	
	\item \emph{Levelset persistence of images.} Given a function 
	$f \colon \clx \to \R$ on a domain $\clx$, classical persistence studies the 
	persistent homology of sublevel sets $f^{-1}(-\infty , \rho]$ for an 
	increasing $\rho$. Levelset persistence~\cite{CarlssonSM09} studies instead 
	the zigzag persistence of of the pre-images of intervals, for appropriate 
	$s_1 \leq s_2 \leq \ldots$,
	\begin{equation}\label{eq:zz_filtration_levelsetzigzag}
		\xymatrix @C-5pt{
		& f^{-1}[s_{i-1} , s_{i}] \ar@{->}[l]|-{\cdots} \ar@{->}^-\subseteq[r] 
		& f^{-1}[s_{i-1} , s_{i+1}] \ar@{<-}^-\supseteq[r] 
		& f^{-1}[s_{i} , s_{i+1}] \ar@{->}[r]|-{\cdots} 
		& \\
		}.
	\end{equation}
	From the levelset persistence, one can recover the sublevel set 
	persistence~\cite{CarlssonSM09}, while maintaining again much smaller 
	structures. This application is of particular importance for medical imaging 
	and material science~\cite{7025987,6873268,6134731}.
\end{enumerate}

\paragraph{Streaming model and memory efficiency.} 
A main advantage of zigzag persistence is to consequently maintain much smaller 
complexes over the computation. To formalize this notion, we adopt a streaming 
model for the computation of zigzag persistence. The input is given by a stream 
of insertions and deletions of faces, with no knowledge of the entire zigzag 
filtration, and zigzag persistence is computed ``on the fly''. In particular, 
the memory complexity of our algorithms, depends solely on the maximal size of 
any complex in the filtration, $\max_i |\clx_i|$, as opposed to the entire 
number of insertions and deletions of faces, which is generally much larger.

\paragraph{Contributions and existing results.} 
In the spirit of~\cite{MischaikowN13}, we introduce a preprocessing reduction of 
a zigzag filtration based on discrete Morse theory~\cite{Forman98DMT}. After 
introducing some background in Section~\ref{sec:background}, we introduce in 
Section~\ref{sec:hl_algo} a {\em zigzag Morse filtration} that generalizes the 
filtered Morse complex~\cite{MischaikowN13} of standard persistence, and we 
prove that it has same persistent homology as the input zigzag filtration. 
Because of removal of cells not agreeing with the Morse decomposition, the 
zigzag Morse filtration contains chain maps that are not inclusions. We study 
the effect of those maps on the boundary operator of the Morse complex in 
Section~\ref{sec:boundary}, and design a persistence algorithm for zigzag Morse 
complexes in Section~\ref{sec:ll_algo}. 
Finally, we report on the experimental performance of the zigzag persistence 
algorithm for Morse complexes in Section~\ref{sec:expe}.

Note that a similar approach to adapt discrete Morse theory to zigzag 
persistence was followed by Escolar and Hiraoka~\cite{EmersonDMTZZ2014}. 
Adapting~\cite{MischaikowN13}, they define a \emph{global} zigzag filtered Morse 
complex for a zigzag filtration, and study its interval decomposition. The main 
limitation of their approach is that the user must know the entirety of the 
input zigzag filtration to compute the Morse pairing, canceling the benefit of 
using ``small complexes'' in zigzag persistence. On the contrary, our approach 
requires no other than local knowledge of the input zigzag filtration, and all 
computation are done ``on the fly'' in the streaming model.

%% file: background.tex
\paragraph{Quiver theory.}
Throughout this article, we fix a field $(\field,+,\cdot)$.
An \emph{$A_n$-type quiver} $\quiv$ is a directed graph:
\[
	\xymatrix @C-5pt{
		\bullet_1 \ar@{<->}[r] 
		& \bullet_2 \ar@{<->}[r] 
		& \cdots \ar@{<->}[r] 
		& \bullet_{n-1} \ar@{<->}[r] 
		& \bullet_n
	},
\]
where, by convention in this article, bidirectional arrows are either forward 
or backward.

An \emph{$\field$-representation} of $\quiv$ is an assignment of a finite
dimensional $\field$-vector space $V_i$ for every node $\bullet_i$ and
an assignment of a linear map $f_i \,\colon\, V_i \leftrightarrow V_{i+1}$ for 
every arrow $\bullet_i \leftrightarrow \bullet_{i+1}$, the orientation of the 
map being the same as that of the arrow. We denote such a representation by 
$\Vmod = (V_i,f_i)$. In computational topology, an $\field$-representation of
an $A_n$-type quiver is called a \emph{zigzag module}.

\begin{wrapfigure}[5]{r}{90pt}
	\vspace{-37pt}
	\[
		\xymatrix{
			V_i \ar@{<->}[r]^{f_i} \ar[d]_{\phi_i} 
			& V_{i+1} \ar[d]^{\phi_{i+1}} \\
			W_{i}\ar@{<->}[r]^{g_i}
			& W_{i+1}
		}
	\]
\end{wrapfigure}

Let $\Vmod = (V_i,f_i)$ and $\Wmod = (W_i, g_i)$ be two 
$\field$-representations of a same quiver $\quiv$. 
A \emph{morphism of representations} $\phi \,\colon\, \Vmod \to \Wmod$ is a 
set of linear maps $\{\phi_i\,:\, V_i \to W_i\}_{i = 1 \ldots n}$ such that 
the diagram on the right commutes for every arrow of~$\quiv$. The morphism is 
called an \emph{isomorphism} (denoted by $\cong$) if every $\phi_i$ is 
bijective.

The \emph{direct sum} of two $\field$-representations $\Vmod = (V_i, f_i)$, 
$\Wmod = (W_i, g_i)$, denoted by $\Vmod \oplus \Wmod$, is the representation 
of $\quiv$ with space $V_i \oplus W_i$ for every node $\bullet_i$, and with 
map $f_i \oplus g_i = \left(
	\begin{smallmatrix}
		f_i	& 0\\
		0	& g_i
	\end{smallmatrix}
\right)$
for every arrow $\bullet_i \leftrightarrow \bullet_{i+1}$. An
$\field$-representation $\Vmod$ is \emph{decomposable} if it can be
written as the direct sum of two non-trivial representations. It is
otherwise said to be \emph{indecomposable}.

Finally, for any $1 \leq b \leq d \leq n$, define the 
\emph{interval representation} $\Imod[b;d]$ as follows:
\[
	\xy
	\xymatrix @C-5pt{
		0 \ar@{<->}[r]^-0 & 
		\cdots \ar@{<->}[r]^-0 & 
		0 \ar@{<->}[r]^-0 & 
		\field \ar@{<->}[r]^-{\id} & 
		\cdots \ar@{<->}[r]^-{\id} & 
		\field \ar@{<->}[r]^-0 & 
		0 \ar@{<->}[r]^-0 & 
		\cdots \ar@{<->}[r]^-0 & 
		0 
	}
	\POS"1,1"."1,3"!C*\frm{_\}},+D*++!U\txt{$\scriptstyle{[1;b-1]}$}
	\POS"1,4"."1,6"!C*\frm{_\}},+D*++!U\txt{$\scriptstyle{[b;d]}$}
	\POS"1,7"."1,9"!C*\frm{_\}},+D*++!U\txt{$\scriptstyle{[d+1;n]}$}
	\endxy
	,
\]
where the maps $0$ and $\id$ stand respectively for the null map and the 
identity map.

Theorem~\ref{thm:zz_gabriel} states that every representation of an 
$A_n$-type quiver can be decomposed into interval representations, which are 
the indecomposables for that quiver:

\begin{theorem}[Krull-Remak-Schmidt, Gabriel]\label{thm:zz_gabriel}
	Every $\field$-representation $\Vmod$ of an $A_n$-type quiver can be 
	decomposed as a direct sum of indecomposables: 
	$\,\Vmod \cong \Vmod^1 \oplus \Vmod^2 \oplus \cdots \oplus \Vmod^N$, where 
	each indecomposable $\Vmod^j$ is isomorphic to some interval 
	representation $\Imod[b_j;d_j]$. This decomposition is unique up to 
	permutation of the indecomposables.
\end{theorem}

In computational topology, such algebraic decomposition of a zigzag module is 
called an \emph{interval decomposition}.

\paragraph{Complexes and homology.}
We refer the reader to~\cite{lefschetz1942algebraic} for an introduction to 
general abstract complexes and their homology, and to~\cite{0025666} for an 
introduction to persistent homology.

Note that, in practice, it is common to work with specific complexes, such as 
simplicial or cubical complexes (as in Section~\ref{sec:expe}). However, Morse 
reductions (introduced below) produce general complexes, which forces us to work
in this general setting.

An \emph{abstract complex} over a principal ideal domain $\PID$ (such as the 
ring of integers $\Z$ or a field $\Z / p\Z$ for $p$ prime) is a graded 
finite collection $\clx = \bigsqcup_{d \in \Z} \clx_d$ of elements, called 
\emph{cells} or \emph{faces}, together with an \emph{incidence function} 
$\inc{\cdot}{\cdot}^{\clx} \colon \clx \times \clx \to \PID$. The 
\emph{dimension} of a cell $\sigma \in \clx_d$ is $\dim \sigma = d$. The 
incidence function satisfies, for any cells $\sigma$, $\tau$, and $\mu$:
\[ 
	\inc{\sigma}{\tau}^{\clx} \neq 0 \Rightarrow \dim \sigma = \dim \tau + 1
	\quad \text{ and } \quad
	\sum_{\tau \in \clx} \inc{\sigma}{\tau}^{\clx} \cdot \inc{\tau}{\mu}^{\clx}
	= 0.
\]
If $\inc{\sigma}{\tau}^{\clx} \neq 0$, we call $\tau$ a \emph{facet} of 
$\sigma$, and $\sigma$ a \emph{cofacet} of $\tau$. If a cell has no cofacet, it
is called \emph{maximal}.

Standard examples of complexes are \emph{simplicial complexes} and 
\emph{cubical complexes}, with an orientation fixed on their cells. In this 
case, the principal ideal domain $\PID$ is the ring of integers $\Z$, and 
incidence function takes values in $\{-1, 0, 1\} \subset \Z$. In this work, we
consider general complexes because they appear under the form of 
\emph{Morse complexes}, defined later.

For a field of coefficients $\field$, we associate to a complex 
$(\clx, \inc{\cdot}{\cdot}^{\clx})$ a \emph{chain complex} 
$\chain(\clx, \field) = \bigoplus_d \chain_d(\clx, \field)$, where 
$\chain_d(\clx, \field)$ is the $\field$-vector space freely generated by the
$d$-dimensional cells $\clx_d$ of $\clx$. For every dimension $d$, the 
\emph{boundary operator} 
$\mbo{\clx}_d \colon \chain_d(\clx) \to \chain_{d-1}(\clx)$ is generated by:
\[
	\mbo{\clx}_d \sigma = 
	\sum_{\tau \in \clx_{d-1}} \inc{\sigma}{\tau}^{\clx} \cdot \tau.
\]
The \emph{$d$-cycles} and \emph{$d$-boundaries} are 
$\cycle_d(\clx,\field) = \ker \mbo{\clx}_d$ and 
$\border_d(\clx,\field) = \im \partial_{d+1}$ respectively, and the 
$d^{th}$ homology group is the quotient
\[
	\homfund{\clx}{d} = \faktor{\cycle_d(\clx,\field)}{\border_d(\clx,\field)}.
\]

In order to simplify notations, we fix the field $\field$ for the rest of the 
article, and remove it from notations. 
To put emphasis on the boundary operator, we denote a complex by $(\clx, \bo)$,
where $\mbo{} \colon C(\clx) \to C(\clx)$ is 
$\bo = \bigoplus_d \, \mbo{\clx}_d$. 
We avoid the superscript $\mbo{\clx}$ when possible.

We denote by 
$\isp{\cdot}{\cdot} \colon \chain(\clx) \times \chain(\clx) \to \Z$ the 
inner product on $\chain(\clx)$ making the canonical basis of cells 
$\{\sigma\}_{\sigma \in \clx}$ orthonormal. In particular, if $\tau$ is in the 
boundary of $\sigma$, $\isp{\partial \sigma}{\tau} = \inc{\sigma}{\tau}^{\clx}$ in 
$(\clx,\mbo{})$. 
For a chain $c \in C(\clx)$, we say that $c$ \emph{contains} a cell $\sigma$, 
and write $\sigma \in c$, if the coefficient of $\sigma$ is non-zero in $c$.

\begin{definition}\label{def:fil}
	Let $\clx$ and $\clx'$ be two complexes; $\clx$ is included in $\clx'$ if 
	$\clx \subseteq \clx'$ as sets of cells, and 
	$\restr{\inc{\cdot}{\cdot}^{\clx'}}{\clx} = \inc{\cdot}{\cdot}^{\clx}$. 
	We also denote the inclusion of complexes by $\clx \subseteq \clx'$.

	A \emph{standard filtration} is a finite collection of complexes with 
	inclusion relations going one way 
	$\clx_1 \subseteq \clx_2 \subseteq \clx_3 \subseteq \cdots$. 
	A \emph{zigzag filtration} is a collection of complexes with inclusion 
	relations going both ways 
	$\clx_1 \subseteq \clx_2 \supseteq \clx_3 \subseteq \cdots$. 
\end{definition}

Finally, a \emph{chain map} $\psi \colon \chain(\clx) \to \chain(\clx')$ is a 
map that commutes with the boundary operators of $\clx$ and $\clx'$. It induces 
a morphism $\psi_* \colon \Hom(\clx) \to \Hom(\clx')$ of homology groups.

\noindent
\begin{minipage}{0.73\textwidth}
	\medskip
	\begin{notation}
		Let $\clx, \clx', Y, Y'$ be complexes, such that $\clx \subseteq \clx'$ 
		and $Y \subseteq Y'$, and let $\phi \colon \chain(\clx) \to \chain(Y)$ 
		and $\phi' \colon \chain(\clx') \to \chain(Y')$ be chain maps. If the 
		square on the right commutes, we allow ourselves to use the same 
		notation $\phi$ for both $\phi$ and $\phi'$, when there is no ambiguity 
		on their domain and codomain.
	\end{notation}
\end{minipage}
\hspace{3px}
\begin{minipage}{0.23\textwidth}
	\flushright
	$	
		\xymatrix@C-10pt @R-10pt{
			C(\clx) \ar@{->}[r]^{\subseteq} \ar[d]_{\phi} 
			& C(\clx') \ar[d]^{\phi'} \\
			C(Y)\ar@{->}[r]^{\subseteq}
			& C(Y')
		}
	$
\end{minipage}

\begin{notation}
	By a small abuse of notations, when two complexes $\clx$ and 
	$\clx \cup \{\sigma\}$ differ by a single cell $\sigma$, we use the notation
	$\xymatrix{\clx \, \ar@{^{(}->}[r]^-{\sigma} & \clx \cup \{ \sigma \}}$ to 
	name the chain map induced by the inclusion. When they differ by a set of 
	cells $\Sigma$, we use the notation 
	$\xymatrix{\clx \, \ar@{^{(}->}[r]^-{\Sigma} & \clx \cup \Sigma}$. 
\end{notation}

\paragraph{Discrete Morse theory.} 
We refer the reader to~\cite{Forman98DMT} for an introduction to discrete Morse
theory, and to~\cite{MischaikowN13} for its application in persistent homology. 
We follow the general presentation of~\cite{MischaikowN13}.

The incidence function of a complex induces a \emph{face partial ordering} $<$ 
on $\clx$ by taking the transitive closure of the relation $\prec$ defined by
\[
	\tau \prec \sigma \quad\text{ iff }\quad \inc{\sigma}{\tau}^{\clx} \neq 0.
\]
A \emph{partial matching} of $\clx$ is a partition 
$\clx = \morsepart$ of the cells of the complex, together with a 
bijective pairing $\Q \leftrightarrow \K$, such that if 
$(\tau, \sigma) \in \Q \times \K$ are paired, then 
$\dim \sigma = \dim \tau + 1$, and $\inc{\sigma}{\tau}^{\clx} \neq 0$ is a unit 
in $\PID$ (e.g., $1$ or $-1$ if $\PID=\Z$). We call such pair of cells a 
\emph{Morse pair}. 
We denote the bijection $\w \colon \Q \to \K$, such that Morse pairs are of the 
form $(\tau, \w(\tau))$.

Call $\hasse$ the \emph{oriented Hasse diagram} of $(\clx, <)$ 
where arrows are oriented downwards (i.e., from higher to lower dimensions), 
except for the arrows between cells of Morse pairs 
$(\tau, \sigma) \in \Q \times \K$, oriented upwards. 

A \emph{Morse matching} of a complex $\clx$ is a partial matching 
that induces an \emph{acyclic} oriented Hasse diagram $\hasse$ for $\clx$.
We denote a Morse matching with a partition $\morsepart$ and pairing 
$\w \colon \Q \to \K$ by $\morse{}$. Note that a Morse matching can also be 
defined on a subset $\Sigma$ of cells of a complex $\clx$. By convention, we 
denote $\morse{}$ Morse matchings for a \emph{complex}, and $\morseh{}$ Morse 
matchings for a \emph{set of faces} not forming a complex.

In a complex with a Morse matching, a \emph{gradient path} between a 
$d+1$-dimensional cell $\nu$ and a $d$-dimensional cell $\mu$ is a simple 
directed path in $\hasse$ from $\nu$ to $\mu$ alternating between $d$ and 
$d+1$-dimensional cells\footnote{
	Note that our definition differs from the original 
	reference~\cite{Forman98DMT}, where gradient paths connect cells of same 
	dimension.
}. 
Every gradient path $\gamma$ is consequently simple and of the form:
\begin{equation}\label{eq:path}
	\gamma = \xymatrix @C-12pt @R-25pt{
		\nu
		&
		& \w(\tau_1)
		&
		& \w(\tau_2) \ \ \ldots \ \ 
		&
		& \w(\tau_{r}) 
		& 
		& \dim d+1 \\
		& \tau_1 \ar@{<-}[lu] \ar@{->}[ru]
		& 
		& \tau_2 \ar@{<-}[lu] \ar@{->}[ru]
		& 
		& \tau_{r} \ar@{<-}[lu] \ar@{->}[ru]
		& 
		& \mu \ar@{<-}[lu]
		& \dim d.
	}
\end{equation}

We denote by $\Gamma(\nu,\mu)$ the set of all distinct gradient paths from $\nu$
to $\mu$, and we define for every path $\gamma$ (with the notations of 
Diagram~(\ref{eq:path})) its \emph{multiplicity} $m(\gamma)$:
\[
	m(\gamma) := \inc{\nu}{\tau_1}^{\clx} \cdot (-1)^r \cdot 
	\prod_{i=1}^{r} \left(\inc{\w(\tau_i)}{\tau_i}^{\clx}\right)^{-1} \ \cdot \ 
	\prod_{i=1}^{r-1} \inc{\w(\tau_i)}{\tau_{i+1}}^{\clx}  \ \cdot \ 
	\inc{\w(\tau_r)}{\mu}^{\clx}
\]
and $m(\gamma) = \inc{\nu}{\mu}^{\clx}$ for the one-edge path 
$\gamma = (\nu,\mu)$, if it exists. In other words, the multiplicity is the 
product of incidences for downward arrows, times the product of minus the 
inverse of incidences for upward arrows in the path.

Given a complex $\clx$ and a Morse matching $\morse{}$, the \emph{Morse complex} 
$(\A, \mbo{\A})$ associated to the matching is the complex based on the cells of 
$\A$, called the \emph{critical cells}, with incidence function 
$\inc{\cdot}{\cdot}^{\A} \colon \A \times \A \to \PID$ defined, for two critical 
cells $\nu, \mu \in \A$, by
\[
	\inc{\nu}{\mu}^{\A} := \sum_{\gamma \in \Gamma(\nu,\mu)} m(\gamma).
\]
The dimension of a critical cell $\sigma$ in $\A$ is the same as the dimension 
of $\sigma$ in the original complex $\clx$. We denote the set of $d$-dimensional 
cells of $\A$ by $\A_d$. As a complex, the boundary operator of $\A$ is defined, 
for $\sigma \in \A_d$ a critical cell of dimension $d$, by
\[
	\mbo{\A}_d : \A_d \to \A_{d-1} ,
	\quad\text{ such that }\quad
	\mbo{\A}_d \tau = \sum_{\mu \in \A_{d-1}} \inc{\nu}{\mu}^{\A} \cdot \mu.
\]

By a small abuse of notation, we refer to $\clx$ and $\A$ as chain complexes 
and write $H(\clx)$ and $H(\A)$ for their homology, provided there is no 
ambiguity in the definition of their incidence function and boundary maps. 

We finally have the fundamental theorem of discrete Morse theory,

\begin{theorem}[Forman~\cite{Forman98DMT}] \label{thm:forman}
	A complex $(\clx,\bo^{\clx})$ and a Morse complex $(\A, \mbo{\A})$, for a 
	Morse matching $\morse{}$ of $\clx$, have isomorphic homology 
	groups\footnote{In fact, the complexes are \emph{homotopy equivalent}.}.
\end{theorem}

\paragraph{Persistent homology and discrete Morse theory.} 
We refer the reader to~\cite{MischaikowN13} for the study of the (standard) 
persistent homology of discrete Morse complexes.

Persistent homology is the study of persistent modules induced by filtrations. 
Let $\clx_1 \subseteq \ldots \subseteq \clx_n$ be a filtration of complexes. A 
\emph{standard Morse filtration} (called \emph{filtered Morse complex} 
in~\cite{MischaikowN13}) for this filtration is a collection of Morse matchings 
$\morse{i}_{i = 1 \ldots n}$ for each $\clx_i$, with Morse complex 
$(\A_i, \mbo{\A_i})$ on the critical cells, and Morse pairs 
$\xymatrix{\omega_i \colon \K_i \ar@{->}[r]|-{\text{bij.}} & \Q_i\\}$, 
satisfying:
\begin{equation}\label{eq:stdmorse}
	\A_i \subseteq \A_{i+1}, \quad \Q_i \subseteq \Q_{i+1},
	\quad \K_i \subseteq \K_{i+1}, \quad \restr{\w_{i+1}}{\Q_i} = \w_i, 
	\quad \restr{\mbo{\A_{i+1}}}{\A_i} = \mbo{\A_{i}}.
\end{equation}
A filtered Morse complex consequently forms a filtration 
$\A_1 \subseteq \ldots \subseteq \A_n$ of Morse complexes connected by 
inclusions. It induces naturally a persistence module:
\[
	\xymatrix @C-5pt{
		\homfun{\A_1} \ar@{->}[r] 
		& \homfun{\A_2} \ar@{->}[r] 
		& \cdots \ar@{->}[r] 
		& \homfun{\A_{n-1}} \ar@{->}[r] 
		& \homfun{\A_n}
	}.
\]

Forman's isomorphism between homology groups of complexes and Morse complexes 
extends to persistent homology groups within this framework. Specifically,

\begin{theorem}[Forman~\cite{Forman98DMT}, 
Mischaikow and Nanda~\cite{MischaikowN13}]
\label{thm:filteredmorse}
	Let $\morse{i}_{i = 1 \ldots n}$ be a standard Morse filtration for a 
	filtration $\clx_1 \subseteq \ldots \subseteq \clx_n$. There exist 
	collections of chain maps 
	$(\psi_i : \chain(\clx_i) \to \chain(\A_i))_{i=1 \ldots n}$ 
	and 
	$(\varphi_i : \chain(\A_i) \to \chain(\clx_i))_{i=1 \ldots n}$ 
	for which the following diagrams commute for every $i$:
	\[
		\xymatrix@R-10pt{
			\chain(\clx_i) \ar@{->}[r]^{\subseteq} \ar[d]_{\psi_i} 
			& \chain(\clx_{i+1}) \ar[d]^{\psi_{i+1}} \\
			\chain(\A_i) \ar@{->}[r]^{\subseteq}
			& \chain(\A_{i+1})
		}
		\hspace{80pt}
		\xymatrix@R-10pt{
			\chain(\clx_i) \ar@{->}[r]^{\subseteq} \ar@{<-}[d]_{\varphi_i} 
			& \chain(\clx_{i+1}) \ar@{<-}[d]^{\varphi_{i+1}} \\
			\chain(\A_i) \ar@{->}[r]^{\subseteq}
			& \chain(\A_{i+1})
		}
	\]
	and $\varphi_i$ and $\psi_i$ induce isomorphisms at the homology level, that 
	are inverses of each other. Consequently, these maps induce 
	isomorphisms between the persistent modules of the filtration and the 
	Morse filtration.
\end{theorem}

Without expressing them explicitly, we use the following properties of the map 
$\psi$ (see~\cite{MischaikowN13} for explicit formulations):

\begin{prop}\label{prop:phi}
	Let $\clx$ be a complex with a Morse matching $\morse{}$. The chain map 
	$\psi \colon \chain(\clx) \to \chain(\A)$ can be expressed as the 
	composition of elementary chain maps over all Morse pairs $(\tau,\sigma)$, 
	taken in an arbitrary order,
	\[
		\psi = \prod_{(\tau,\sigma), \text{ s.t. } \sigma = \w(\tau)} 
		\psi_{\tau,\sigma}\,,
	\]
	where $\psi_{\tau,\sigma} \colon 
	\chain(X') \to \chain(X'\setminus \{\tau,\sigma\})$ is defined on a 
	``partially reduced'' complex $\clx'$ to $\clx' \setminus \{\tau,\sigma\}$, 
	with incidence functions induced by the partial matching. More specifically, 
	$\clx'$ is a Morse complex of $\clx$ for a matching 
	$(\A',\allowbreak \Q',\allowbreak \K', \w')$, such that $\Q' \subseteq \Q$, 
	$\K' \subseteq \K$, and the restriction of $\w$ to $\Q'$ is equal to $\w'$. 
	The complex $\clx' \setminus \{\tau,\sigma\}$ is the Morse complex of $\clx$ 
	with one more Morse pair $(\tau,\sigma)$. The set of Morse pairs already 
	considered in $\Q' \times \K'$ is dependent of the order in which the maps 
	are composed. 

	The map $\psi_{\tau,\sigma}$ satisfies:
	\begin{enumerate}
		\item $\psi_{\tau,\sigma} (\sigma) = 0$,
		\item $\psi_{\tau,\sigma} (\tau)$ is a linear combination of facets 
		of $\sigma$ in $\clx'$, and
		\item $\psi_{\tau,\sigma} (\mu) = \mu$ for all 
		$\mu \neq \sigma, \tau$.
	\end{enumerate} 

	Similarly, the map $\varphi \colon \chain(\A) \to \chain(\clx)$ can be 	
	decomposed into
	\[
		\varphi = \prod_{(\tau,\sigma), \text{ s.t. } \sigma = \w(\tau)} 
		\varphi_{\tau,\sigma} \,,
	\]
	such that $\varphi_{\tau,\sigma} \colon 
	\chain(X'\setminus \{\tau,\sigma\}) \to \chain(X')$ and 
	$\psi_{\tau,\sigma} \colon 
	\chain(X') \to \chain(X'\setminus \{\tau,\sigma\})$ induce isomorphisms at 
	the homology level, that are inverse of each other (defined on the 
	appropriate domain and codomain).
\end{prop}

%% file: hl_algo.tex
For a zigzag filtration of complexes $\F$, we introduce in this article a 
canonical zigzag filtration $\M$ of Morse complexes admitting the same 
persistent homology. 

\subsection{Zigzag Morse filtration}\label{sec:zzmorsefil}

Without loss of generality, consider the zigzag filtration 
\begin{equation}\label{eq:genzzfil}
	\oF \ := \ \xymatrix @C+8pt{
		 \oclx_1 \ar@{^{(}->}^-{\Sigma_1}[r] 
		& \oclx_2 \ar@{<-^{)}}^-{\Sigma_2}[r] 
		& \enspace \cdots \enspace \ar@{^{(}->}^-{\Sigma_{2k-1}}[r] 
		& \oclx_{2k-1} \ar@{<-^{)}}^-{\Sigma_{2k}}[r] 
		& \oclx_{2k} \\
  	},
\end{equation}
where the $\oclx_i$ are complexes, $\oclx_1 = \oclx_{2k} = \emptyset$, and the 
$i^{th}$ arrow is an inclusion, either forward ($i$ odd) or backward ($i$ even), 
where complexes $\oclx_i$ and $\oclx_{i+1}$ differ by a set of cells $\Sigma_i$ 
(possibly empty). We now further decompose $\oF$.

\paragraph{Atomic operations.} 
For each forward arrow 
$\xymatrix@C-5pt{\bullet_{i} \ar@{->}[r] & \bullet_{i+1}}$, $i$ odd, let 
$\morseh{i}$ be a Morse matching of the set of cells $\Sigma_i$.

Because Morse matchings are acyclic, there exists a total ordering of the cells 
of $\Sigma_i$, compatible with the face partial ordering of $\Sigma_i$, such 
that paired cells in $\morseh{i}$ are consecutive with regard to that order. We 
can consequently decompose a forward inclusion $\oclx_i \subseteq \oclx_{i+1}$ 
into a sequence of inclusions of a single critical cell $\sigma \in \Ah_i$, and 
of inclusions of a single Morse pair of cells 
$(\tau,\sigma) \in \Qh_i \times \Kh_i$, with $\sigma = \wh_i(\tau)$.

For every backward arrow 
$\xymatrix@C-5pt{\bullet_{i} \ar@{<-}[r] & \bullet_{i+1}}$, $i$ even, the Morse 
matchings $\morseh{j}$, for smaller odd indices $j < i$, induce a Morse matching 
on the cells of $\clx_i$. To avoid ambiguity, if a cell is reinserted 
in the filtration after being removed it is considered as a different element.
By restriction, they consequently induce a valid Morse 
matching on all cells of $\Sigma_i$, except on those cells $\sigma \in \Sigma_i$ 
that form a Morse pair $(\tau,\sigma)$, with $\tau \notin \Sigma_i$. We 
decompose backward arrows into a sequence of removals of a single critical cell, 
of removals of a single Morse pair of cells, and of removals of a non-critical 
cell $\sigma$, without its paired cell $\tau \notin \Sigma_i$.

In summary, given an input filtration $\oF$ as above, and the Morse matchings 
$\morseh{i}$, we defined an {\em atomic} zigzag filtration 
\[
	\F \ := \ \ \ \xymatrix @C+8pt{
  		(\emptyset =) \ \clx_1 \ar@{<->}[r] 
  		& \clx_2 \ar@{<->}[r] 
  		& \cdots \ar@{<->}[r] 
  		& \clx_{n-1} \ar@{<->}[r] 
  		& \clx_n \ (= \emptyset) \\
  	},
\]
where all arrows are of the following three types:

\hspace{-0.9cm}
\begin{minipage}{0.26\textwidth}
	\begin{equation}\label{eq:type1}
		\xymatrix@C+5pt{
			\clx \ar@{<->}[r]^-{\sigma} 
			& \clx'
		} 
	\end{equation}
\end{minipage}
\hspace{4px}
\begin{minipage}{0.26\textwidth}
	\vspace{-5px}
	\begin{equation}\label{eq:type2}
		\xymatrix@C+5pt{
			\clx \ar@{<->}[r]^-{\{\tau,\sigma\}} 
			& \clx'
		} 
	\end{equation}
\end{minipage}
\hspace{4px}
\begin{minipage}{0.45\textwidth}
	\vspace{9px}
	\begin{equation}\label{eq:type3}
		\xymatrix@C+5pt{
			\clx \ar@{->}[r]^-{\id} 
			& \clx \ar@{<-^{)}}[r]^-{\sigma} 
			& \clx \setminus \{\sigma\}
		}
	\end{equation}
\end{minipage}

\medskip
\noindent
where $\sigma$ is in each case a maximal cell in $X$, Diagrams~(\ref{eq:type1}) 
and~(\ref{eq:type2}) are forward or backward insertions of a critical cell or a 
Morse pair $(\tau,\sigma)$ of cells, respectively, and Diagram~(\ref{eq:type3}) 
is the removal of the cell $\sigma$ from a Morse pair $(\tau,\sigma)$, where the 
cell $\tau$ is not removed. The identity arrow in this last diagram is a 
technicality that is clarified later. Naturally, one can recover the persistent 
homology of the zigzag filtration $\overline{\F}$ from the one of $\F$. 
We work with $\F$ for the rest of the article.

\paragraph{Morse filtration.}
Given a zigzag filtration $\oF$, Morse matchings $\morse{i}$, and an associated 
atomic filtration $\F$ as above, we define a \emph{zigzag Morse filtration}
\[
	\M \ := \quad 
	\xymatrix @C+8pt{
  		(\emptyset =) \,\A_1 \ar@{<->}[r] 
  		& \A_2 \ar@{<->}[r] 
  		& \cdots \ar@{<->}[r] 
  		& \A_{n-1} \ar@{<->}[r] 
  		& \A_n \,(= \emptyset)\\
  	},
\]
of Morse complexes $(\A_i, \mbo{\A_i})$ of the complexes $(\clx_i,\mbo{\clx_i})$ 
of $\F$ inductively. Note that the maps of the zigzag Morse filtration are not 
all inclusions. Specifically, for a critical cell $\sigma$ in both $\clx_i$ and 
$\clx_{i+1}$, in general $\mbo{\A_i}(\sigma) \neq \mbo{\A_{i+1}}(\sigma)$. 

All $\clx_1, \clx_n, \A_1$ and $\A_n$ are empty complexes. The zigzag Morse 
filtration is constructed inductively for the insertion of a critical cell 
(Diagram~(\ref{eq:type1})) and the insertion of a Morse pair 
(Diagram~(\ref{eq:type2})) as for standard Morse 
filtrations~\cite{MischaikowN13}:
\begin{equation}\label{eq:simplecases}
\begin{gathered}
	\xymatrix @C-5pt @R-10pt{
  		\chain(\clx) \ar@{^{(}->}[r]^-{\sigma'} \ar@{->}[d]_-{\psi}
  		& \chain(\clx \cup \{\sigma'\}) \ar@{->}[d]^-{\psi} \\
  		\chain(\A) \ar@{^{(}->}[r]^-{\sigma'}
  		& \chain(\A \cup \{\sigma'\}) \\
	}
	\hspace{2cm}
	\xymatrix @C-5pt @R-10pt{
  		\chain(\clx) \ar@{^{(}->}[r]^-{\{\tau,\sigma\}} \ar@{->}[d]_-{\psi}
  		& \chain(\clx \cup \{\tau,\sigma\}) 
  			\ar@{->}[d]^-{\psi_{\tau,\sigma} \circ \psi}\\
  		\chain(\A) \ar@{->}[r]^-{\id}
  		& \chain(\A)\enspace, \\
	}
\end{gathered}
\end{equation}
where all horizontal arrows are inclusions of complexes, and in particular the 
boundary maps of $\A$ and $\A \cup \{\sigma'\}$ are equal when restricted to the 
cells of $\A$. The removal of critical cells and Morse pairs is symmetrical. The 
chain maps $\psi$ and $\psi_{\tau,\sigma}$ are the ones of 
Theorem~\ref{thm:filteredmorse} and Properties~\ref{prop:phi}, and are used later. 

For the removal of a non-critical cell $\sigma$ without its paired cell $\tau$ 
(Diagram~(\ref{eq:type3})), which is specific to zigzag persistence, the Morse 
filtration is constructed with:
\begin{equation}\label{eq:hardcase}
\begin{gathered}
  	\xymatrix @C-5pt @R-10pt{
    	\chain(\clx) \ar@{->}[r]^-{\id} 
    		\ar@{->}[d]_-{\psi_{\tau,\sigma} \circ \psi}
    	& \chain(\clx) \ar@{<-^{)}}[r]^-{\sigma} \ar@{->}[d]^-{\psi}
    	& \chain(\clx \setminus \{\sigma\}) \ar@{->}[d]^-{\psi} \\
    	\chain(\A,\bo) \ar@{->}[r]^-{\varphi_{\tau,\sigma}}
    	& \chain(\A \cup \{\tau,\sigma\}, \bo') \ar@{<-^{)}}[r]^-{\sigma}
    	& \chain(\A \cup \{\tau\}, \bo'') \enspace . \\
  	}
\end{gathered}
\end{equation}
The main technicality is that the boundary maps $\bo$ and $\bo'$ differ in a non 
trivial way, that we study in Section~\ref{sec:boundary}. The map $\bo''$ is 
equal to the restriction of $\bo'$ to the critical cells $\A \cup \{\tau\}$ (the 
right arrow is a backward inclusion of complexes). The chain maps 
$\psi_{\tau,\sigma}$ and $\varphi_{\tau,\sigma}$ are the ones from 
Theorem~\ref{thm:filteredmorse} and Properties~\ref{prop:phi}, and $\psi$ is the 
compositions of all maps $\psi_{\mu, \w(\mu)}$ over the Morse pairs $(\mu,\w(\mu))$ 
of the Morse matching of $\clx$, except the pair $(\tau,\sigma)$. We give an example 
of zigzag Morse filtration in Figure~\ref{fig:easy_ex}.

\begin{figure}[t]
	\centering
	\includegraphics[width=12cm]{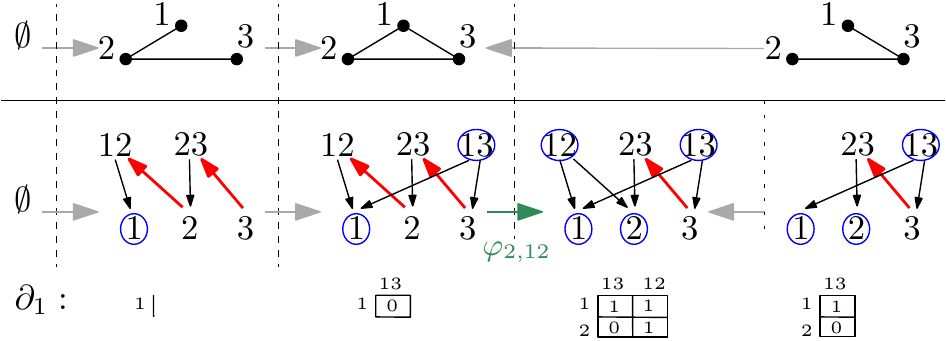}
	\caption{Zigzag filtration (top) and its Morse filtration (bottom), given by 
	Hasse diagrams and (Morse) boundary maps. Upward arrows in Hasse diagrams 
	represent Morse matchings, critical faces are circled. Note that the rightmost 
	operation illustrates Diagram~(\ref{eq:hardcase}), with a non trivial 
	modification of $\partial_1(\{1,3\})$.}
	\label{fig:easy_ex}
\end{figure}

Diagrams~(\ref{eq:simplecases}) are studied in~\cite{MischaikowN13}. We now 
focus on the study of Diagram~(\ref{eq:hardcase}).

\begin{remark}
	Note that a key point for the proofs of theorems in~\cite{MischaikowN13} 
	is that filtered Morse complexes in standard persistence satisfy 
	$(\A_i,\bo) \subset (\A_{i+1},\bo)$. This fact also allows the standard 
	persistent homology algorithm~\cite{EdelsbrunnerLZ02,ZomorodianC05} to 
	work directly for filtered Morse complexes. This property is not satisfied 
	by zigzag Morse filtrations, which explains why our approach is more 
	atomic than the one of~\cite{MischaikowN13} (see 
	Section~\ref{sec:isomodules}), and that we have to design a new homology 
	matrix algorithm to implement operation~(\ref{eq:hardcase}) (see 
	Sections~\ref{sec:boundary} and~\ref{sec:ll_algo}).
\end{remark}

\subsection{Isomorphism of zigzag modules}\label{sec:isomodules}

Theorem~\ref{thm:filteredmorse} implies that the atomic operations of 
Diagrams~(\ref{eq:simplecases}) induce commuting diagrams in homology, with 
vertical maps being isomorphisms as proved in~\cite{MischaikowN13}: 

\begin{lemma} \label{lem:simple_cases}
	Let $\clx$ be a complex and $\morse{}$ a Morse complex obtained 
	from $\clx$. Let $\sigma'$ be a cell, and $(\tau,\sigma)$ a Morse pair, 
	such that $(\A \cup \{\sigma'\},\allowbreak \Q,\allowbreak \K)$ and 
	$(\A,\allowbreak \Q \cup \{\tau\},\allowbreak \K \cup \{\sigma\})$ are 
	valid Morse complexes. Then there exist isomorphisms $\psi_*$ and 
	$(\psi_{\tau,\sigma})_*$ such that the following diagrams commute:
	\[
		\xymatrix @C-5pt @R-10pt{
			\hf{\clx} \ar@{->}[r]^-{\sigma'_*} \ar@{->}[d]_-{\psi_*}
			& \hf{\clx \cup \{\sigma'\}} \ar@{->}[d]^-{\psi_*} \\
			\hf{\A} \ar@{->}[r]^-{\sigma'_*}
			& \hf{\A \cup \{\sigma'\}} \\
		}
		\hspace{2cm}
		\xymatrix @C-5pt @R-10pt{
			\hf{\clx} \ar@{->}[r]^-{\sigma_* \,\circ\, \tau_*} 
				\ar@{->}[d]_-{\psi_*}
			& \hf{\clx \cup \{\tau,\sigma\}} 
				\ar@{->}[d]^-{(\psi_{\tau,\sigma})_* \circ \psi_*} \\
			\hf{\A} \ar@{->}[r]^-{\id}
			& \hf{\A} \\
		}
	\]
	where $\sigma'_*$ and $\sigma_* \circ \tau_*$ are the maps induced at 
	homology level by the insertion of $\sigma'$ and $\{\tau,\sigma\}$ 
	respectively.
	The maps $\psi_*$ and $(\psi_{\tau,\sigma})_*$ are the isomorphisms 
	induced by chain maps $\psi$ and $\psi_{\tau,\sigma}$ of discrete Morse theory 
	(see Theorem~\ref{thm:filteredmorse}).
\end{lemma}

We prove the following lemma, which is specific to our zigzag Morse filtration.

\begin{lemma} \label{lem:rm_half_pair}
	Let $\clx$ be a complex and $\morse{}$ a Morse complex obtained 
	from $\clx$. Let $\sigma$ be a maximal cell of $\clx$ not in $\A$, which 
	therefore forms a Morse pair with a cell $\tau$, 
	$\inc{\sigma}{\tau}^\clx \neq 0$. There exist isomorphisms $\psi_*$, 
	$(\psi_{\tau,\sigma})_*$, and $(\varphi_{\tau,\sigma})_*$ such that the 
	following diagram commutes:
	\[
		\xymatrix @C-5pt @R-10pt{
			\hf{\clx} \ar@{->}[r]^-{\id} 
				\ar@{->}[d]_-{(\psi_{\tau,\sigma})_* \circ \psi_*}
			& \hf{\clx} \ar@{<-}[r]^-{\sigma_*} \ar@{->}[d]^-{\psi_*}
			& \hf{\clx \setminus \{\sigma\}} \ar@{->}[d]^-{\psi_*} \\
			\hf{\A} \ar@{->}[r]^-{(\varphi_{\tau,\sigma})_*}
			& \hf{\A \cup \{\tau,\sigma\}} \ar@{<-}[r]^-{\sigma_*}
			& \hf{\A \cup \{\tau\}} \\
		}
	\]
	where $\sigma_*$ is the map induced at homology level by the removal of 
	$\sigma$. 
	The maps $\psi_*$, $(\psi_{\tau,\sigma})_*$, and 
	$(\varphi_{\tau,\sigma})_*$ are the isomorphisms induced at homology level 
	by, respectively, the chain maps $\psi$, $\psi_{\tau,\sigma}$, and 
	$\varphi_{\tau,\sigma}$ of discrete Morse theory 
	(see Theorem~\ref{thm:filteredmorse}).
\end{lemma}

\begin{proof}
	Apply the homology functor to Diagram~(\ref{eq:hardcase}). The right square 
	commutes, being induced by horizontal inclusions. Because the maps induced at 
	homology level by $\psi_{\tau,\sigma}$ and $\varphi_{\tau,\sigma}$ are 
	isomorphisms, inverse of each other (see Theorem~\ref{thm:filteredmorse}), we get
	$
		(\varphi_{\tau,\sigma})_* \circ (\psi_{\tau,\sigma})_* \circ \psi_* 
		= \psi_* 
	$
	and the left square commutes.
\end{proof}

We conclude,

\begin{theorem} \label{thm:equivalence}
	The zigzag filtrations $\F$ and $\M$ have same persistent homology.
\end{theorem}

\begin{proof}
	Applying the homology functor to $\F$ and $\M$, we get the zigzag modules
	\[
		\xymatrix @C-5pt @R-10pt{
			\hf{\F} \,:
			& \hf{\clx_0} \ar@{<->}[r] \ar@{->}[d]^-{\psi_*^0}
			& \hf{\clx_1} \ar@{<->}[r] \ar@{->}[d]^-{\psi_*^1}
			& \cdots \ar@{<->}[r]
			& \hf{\clx_m} \ar@{->}[d]^-{\psi_*^m} \\
			\hf{\M} \,:
			& \hf{\A_0} \ar@{<->}[r]
			& \hf{\A_1} \ar@{<->}[r]
			& \cdots \ar@{<->}[r]
			& \hf{\A_m} \\
		}
	\]
	where, by construction, every $\A_i$ is a Morse complex of $\clx_i$, and the 
	$\psi_*^i$ are the isomorphisms induced by the chain maps 
	$\psi^i \colon \chain(\clx_i) \to \chain(\A_i)$, connecting a complex and 
	its Morse reduction (Theorem~\ref{thm:filteredmorse}). 
  	By Theorem~\ref{thm:filteredmorse} and Lemma~\ref{lem:rm_half_pair}, all 
  	squares commute and are compatible with each other, and the $(\psi_*^i)$ 
  	define an isomorphism of zigzag modules.
\end{proof}

%% file: boundary.tex
Referring to Diagram~(\ref{eq:hardcase}), let $\clx$ be a complex with incidence 
function $\inc{\cdot}{\cdot}^\clx$, together with a Morse matching $\morse{}$, 
inducing an orientation of the Hasse diagram $\hasse$ of the complex, and a 
Morse complex $(\A, \partial)$. 

In this section, we track the evolution of the boundary operators in 
Morse complexes under the evaluation of the map 
$\varphi_{\tau,\sigma} \colon (\A, \bo) \to (\A \cup \{\tau,\sigma\}, \bo')$ 
from Diagram~(\ref{eq:hardcase}). Both complexes are Morse complexes of the same 
$\clx$, whose matchings differ by exactly one pair $(\tau,\sigma)$, i.e., 
the Morse partition of complex $\A \cup \{\tau, \sigma\}$ is 
$(\A \cup \{\tau,\sigma\}) \sqcup (\Q \setminus \{\tau\}) \sqcup 
(\K \setminus \{\sigma\})$. We denote this last complex by $(\A',\bo')$, with 
incidence function $\inc{\cdot}{\cdot}^{\A'}$ in the following. 
We prove:

\begin{lemma} \label{lem:nbo}
	Let $\nu$ be a cell of the complex $(\A,\bo)$. Then, in the complex 
	$(\A', \bo')$,  
	\begin{equation}
    	\bo'(\nu) = 
	      	\bo(\nu) + 
      		\left(\inc{\sigma}{\tau}^{\clx}\right)^{-1} \inc{\nu}{\tau}^{\A'} 
      		\cdot \bo' \sigma.
  	\end{equation}
\end{lemma}

\begin{proof}
	First, note that $\sigma$ is maximal in $\clx$, and so it is maximal in 
	$\A \cup \{\tau,\sigma\}$.

	Let $\hasse$ and $\hasse'$ be the Hasse diagrams of $\clx$ induced by the 
	Morse matchings of $\A$ and $\A'$, respectively. Because the matchings differ 
	by a single Morse pair $(\tau,\sigma)$, $\hasse$ and $\hasse'$ only differ by 
	the orientation of the edge $\tau \leftrightarrow \sigma$. 

	For a critical cell $\nu \in \A$, we have:
	\[
		\bo \nu  
		= \sum_{\substack{\mu \in \A \\ \gamma \in \Gamma(\nu,\mu)}} 
			m(\gamma) \cdot \mu  
		= \underbrace{
			\sum_{
				\substack{
					\mu \in \A, \\ 
					\gamma \in \Gamma_{\tau \rightarrow \sigma}(\nu,\mu)
				}
			} m(\gamma) \cdot \mu
		}_{(\star)} \ + 
		\underbrace{
			\sum_{
				\substack{
					\mu \in \A, \\ 
					\gamma \in \Gamma_{\tau \nrightarrow \sigma}(\nu,\mu)
				}
			} m(\gamma) \cdot \mu
		}_{\bo' \nu - \inc{\nu}{\tau}^{\A'} \cdot \tau},
	\]
	where $\Gamma_{\tau \rightarrow \sigma}(\nu,\mu)$ are the gradient paths 
	from $\nu$ to $\mu$ in $\hasse$ containing the upward arrow 
	$\tau \rightarrow \sigma$, and 
	$\Gamma_{\tau \nrightarrow \sigma}(\nu,\mu)$ are the ones not containing it.
	Assume $\tau$ is of dimension $d$, and $\sigma$ of dimension $d+1$. 

	Because $\sigma$ is critical in $\A'$, it has no ingoing arrow from cells of 
	dimension $d$ in $\hasse'$. Consequently, 
	$\Gamma_{\tau \nrightarrow \sigma}(\nu,\mu)$ contains exactly all gradient 
	paths from $\nu$ to $\mu \neq \tau$ in $\hasse'$. Hence, the 
	sum over $\Gamma_{\tau \nrightarrow \sigma}(\nu,\mu)$, for $\mu \in \A$, 
	gives $\bo' \nu - \inc{\nu}{\tau}^{\A'} \tau$. Note that $\sigma$ 
	cannot appear in $\bo' \nu$ because $\sigma$ is maximal by hypothesis.  

	Now, studying the left term $(\star)$, and splitting gradient paths passing 
	through edge $(\tau,\sigma)$, then factorizing, we get
	\begin{align*}
		(\star) 
		&= \sum_{
			\substack{
				\mu \in \A, \\ 
				\gamma_1 \in \Gamma(\nu,\tau), \\ 
				\gamma_2 \in \Gamma(\sigma,\mu)
			}
		} m(\gamma_1) \cdot 
			\left(-\inc{\sigma}{\tau}^{\clx}\right)^{-1} m(\gamma_2) \,\cdot\, \mu \\
		&= -\left(\inc{\sigma}{\tau}^{\clx}\right)^{-1} 
		\underbrace{
			\sum_{\mu \in \A} \left( 
				\sum_{\gamma_2 \in \Gamma(\sigma,\mu)} m(\gamma_2) \cdot \mu 
			\right)
		}_{(\star_2) 
			\ = \ \bo' \sigma - \inc{\sigma}{\tau} \cdot \tau
		} \cdot 
		\underbrace{
			\left(\sum_{\gamma_1 \in \Gamma(\nu,\tau)} m(\gamma_1) \right)
		}_{(\star_1)}.
	\end{align*}
	The sum $(\star_1)$ over $\Gamma(\nu,\tau)$ is independent of $\mu$, 
	and equal to $\inc{\nu}{\tau}^{\A'}$ by definition. 

	Because $\tau$ is critical in $\A'$, it has no outgoing arrow towards cells 
	of dimension $d+1$ in $\hasse'$. Consequently, $\Gamma(\sigma, \mu)$ 
	contains exactly all gradient paths from $\sigma$ to $\mu$ in $\hasse'$, 
	where $\mu \neq \tau$. Hence, the sum $(\star_2)$ over $\Gamma(\sigma, \mu)$ 
	gives $\bo' \sigma - \inc{\sigma}{\tau}^{\clx} \cdot \tau$. 

	Finally, putting terms together, the following allows us to conclude:
	\begin{align*}
  		\bo \nu 
  			&= \left(\bo' \nu - \inc{\nu}{\tau}^{\A'} \cdot \tau\right) 
  				- \frac{\inc{\nu}{\tau}^{\A'}}{\inc{\sigma}{\tau}^{\clx}} 
  				\left(\bo' \sigma - \inc{\sigma}{\tau}^{\clx} \cdot \tau \right)\\
			&= \bo' \nu 
				- \left(\inc{\sigma}{\tau}^{\clx}\right)^{-1} 
				\inc{\nu}{\tau}^{\A'} \bo' \sigma.
	\end{align*}
\end{proof}

%% file: ll_algo.tex
We describe in this section, our implementation of the algorithm to compute the persistence diagram of a zigzag Morse filtration as defined in Section~\ref{sec:hl_algo}. It consists of adapting the zigzag persistence algorithm~\cite{MariaO15}, used in our experiments, to our Morse framework, relying on the results of Sections~\ref{sec:hl_algo} and~\ref{sec:boundary}. Our approach could be adapted for implementing algorithm~\cite{CarlssonS10,CarlssonSM09}.

\subsection{Zigzag Persistence algorithm}\label{ssec:stand_alg}

We first recall the algorithms for computing zigzag persistence.

\paragraph{Existing zigzag persistence algorithms.}
There are currently two practical%
\footnote{
	Putting aside~\cite{MilosavljevicMS11}, which is essentially of theoretical 
	nature.
} 
approaches to compute zigzag persistent 
homology~\cite{CarlssonS10,CarlssonSM09,MariaO15}. They can both  be formulated 
in a unified framework~\cite{MariaO16}. Given an input zigzag filtration:
\begin{equation}\label{eq:zz_filtration_2}
	\xymatrix @C-5pt{
		\clx_1 \ar@{->}^-\subseteq[r] 
		& \clx_2 \ar@{<-}^-\supseteq[r] 
		& \cdots \ar@{->}^-\subseteq[r] 
		& \clx_{n-1} \ar@{<-}^-\supseteq[r] 
		& \clx_n
	},
\end{equation}
both algorithms are iterative. At step $i$ of the computation, they maintain a 
homology basis of $H(\clx_i)$ that is \emph{compatible} (defined later) with the 
interval decomposition of the zigzag module associated to a zigzag filtration of the 
form
\begin{equation}\label{eq:zz_filtration_algo}
	\xymatrix @C-5pt{
		\clx_{1} 
		& \clx_{2} \ar@{<->}[l] \ar@{<->}[r] 
		& \cdots 
		& \clx_{i} \ar@{<->}[l] 
		& \ar@{<->}[l] \clx'_{i+1}  
		& \cdots \ar@{<->}[l] 
		& \clx'_{i+m-1} \ar@{<->}[l] \ar@{<->}[r] 
		& \clx'_{i+m}
	},
\end{equation}
The first $i$ complexes and $i-1$ maps in~(\ref{eq:zz_filtration_2}) 
and~(\ref{eq:zz_filtration_algo}) are identical, and the remaining complexes 
and maps of~(\ref{eq:zz_filtration_algo}) are algorithm dependent. 
Both algorithms consist of updating a homology basis in order to maintain its 
compatibility when operating (a subset of) the following three local 
transformations of the zigzag filtration/module in sequence:

\begin{center}
\begin{minipage}{0.48\textwidth}
	\begin{equation}\label{eq:reflection_diamond}
	\begin{gathered}
		\xymatrix @=20pt @C=15pt @R-2pc{
			& \clx\cup\{\sigma\}
			& \\
			\leftrightarrow 
			\clx \ar@{->}[dr]_-{\id} \ar@{^{(}->}[ur]^-{\sigma} 
			&  
			& \clx \ar@{->}[dl]^-{\id} \ar@{_{(}->}[ul]_-{\sigma} 
			\leftrightarrow, \\
			& \clx
			& \\
		}
	\end{gathered}
	\end{equation}
\end{minipage}
\begin{minipage}{0.48\textwidth}
	\begin{equation}\label{eq:reflection_diamond_inv}
	\begin{gathered}
		\xymatrix @=20pt @C=15pt @R-2pc{
			& \clx \cup \{\sigma\}
			& \\
			\leftrightarrow 
			\clx \ar@{<-}[dr]_-{\id} \ar@{<-^{)}}[ur]^-{\sigma} 
			&  
			& \clx \ar@{<-}[dl]^-{\id} \ar@{<-_{)}}[ul]_-{\sigma} 
			\leftrightarrow, \\
			& \clx
			& \\
		}
	\end{gathered}
	\end{equation}
\end{minipage}
\begin{minipage}{0.6\textwidth}
	\begin{equation}\label{eq:transposition_diamond}
	\begin{gathered}
		\xymatrix @=20pt @C=15pt @R-2pc{
			& \clx \cup \{\tau\}
			& \\
			\leftrightarrow 
			\clx\cup\{\sigma,\tau\} \ar@{<-^{)}}[ur]^-{\sigma} 
				\ar@{<-^{)}}[dr]_-{\tau} 
			& 
			& \clx \ar@{_{(}->}[ul]_-{\tau} \ar@{_{(}->}[dl]^-{\sigma} 
			\leftrightarrow, \\
			& \clx \cup \{\sigma\}
			& \\
		}
	\end{gathered}
	\end{equation}
\end{minipage}
\end{center}

\medskip
\noindent
where each arrow represents the insertion of a cell. These transformations 
are called \emph{reflection diamonds} for~(\ref{eq:reflection_diamond}) 
and~(\ref{eq:reflection_diamond_inv}), and \emph{transposition diamonds} 
for~(\ref{eq:transposition_diamond}), and their effect on the interval 
decomposition of the zigzag module have been characterized for general zigzag 
filtrations of complexes in~\cite{MariaO16,MariaO15}. 

We now focus on the algorithm introduced in~\cite{MariaO15} that we use in our experiments.

\paragraph{The zigzag algorithm of~\cite{MariaO15}.}
Let $\F \,:\, 
\xymatrix @C-5pt{
	\clx_{1}
	& \clx_{2} \ar@{<->}[l] \ar@{<->}[r] 
	& \cdots 
	& \clx_{n} \ar@{<->}[l]
}$ be the input zigzag filtration, where \emph{all} arrows are forward or backward inclusions of a single cell. Let $\F_j$ be:
\[
	\xymatrix @C-7pt{
		\clx_{1}
		& \clx_{2} \ar@{<->}[l] \ar@{<->}[r] 
		& \cdots 
		& \clx_{j} \ar@{<->}[l]
		& \clx'_{j+1} \ar@{->}[l]_-{\sigma_1}
		& \cdots \ar@{->}[l]_-{\sigma_2}
		& \clx'_{j+m-1} \ar@{->}[l]_-{\sigma_{m-1}}
		& \clx'_{j+m} = \emptyset \ar@{->}[l]_-{\sigma_m}.
	}
\]
For indices $1 \leq p \leq q \leq n$, denote by $\mathcal{Z}[p;q]$ the restriction of a filtration $\mathcal{Z}$ to spaces of indices $i \in [p;q]$, and maps between them.

Passing from filtration $\F_j$ to filtration $\F_{j+1}$ using reflection and 
transposition diamonds consists of the following:

\begin{enumerate}[itemsep=\medskipamount]
	\item If $\xymatrix{\clx_{j}\ar[r]^-{\sigma} &\clx_{j+1}}$ is forward in 
	$\F$, define $\F_{j+1}$ to be
  	\[
		\xymatrix @C-7pt{
			\clx_{1} \ar@{<->}[r] 
			& \cdots 
			& \clx_j \ar@{<->}[l] \ar@{->}[r]^-{\sigma}
			& \clx_{j+1} \ar@{<-}[r]^-{\sigma}
			& \clx_{j} 
			& \clx'_{j+1} \ar@{->}[l]_-{\sigma_1}
			& \cdots \ar@{->}[l]_-{\sigma_2}
			& \clx'_{j+m} = \emptyset \ar@{->}[l]_-{\sigma_m}
		}.
  	\]
	Considering $\oF_j$ to be $\F_j$ with two extra identity arrows,
  	\[
  		\oF_j :\,
		\xymatrix @C-7pt{
			\clx_{1} \ar@{<->}[r] 
			& \cdots 
			& \clx_j \ar@{<->}[l] \ar@{->}[r]^-{\id}
			& \clx_{j} \ar@{<-}[r]^-{\id}
			& \clx_{j} 
			& \clx'_{j} \ar@{->}[l]_-{\sigma_1}
			& \cdots \ar@{->}[l]_-{\sigma_2}
			& \clx'_{j+m} = \emptyset \ar@{->}[l]_-{\sigma_m}
		},
  	\]
	we have that $\oF_j$ and $\F_{j+1}$ are related by a reflection diamond 
	(Diagram~(\ref{eq:reflection_diamond})) at $\clx_j$. Studying the effect of 
	a reflection diamond on homology, algorithm~\cite{MariaO15} updates a 
	homology matrix (defined below in this framework) at $\clx_j$, compatible 
	with $\F_j$ (and also $\oF_j$), into a homology matrix at $\clx_{j+1}$, 
	compatible with $\F_{j+1}$ defined above.
	
	\item If $\xymatrix{\clx_{j} \ar@{<-}[r]^-{\sigma} &\clx_{j+1}}$ is backward 
	in $\F$, there exists an index $\ell$ such that $\sigma = \sigma_{\ell}$ in 
	the part $\F_j[j;j+m]$ of the filtration $\F_j$. Define $\F_{j+1}$ to be
	\[
  		\xymatrix @-7pt{
			\cdots \clx_{j} \ar@{<-}[r]^-{\sigma_{\ell} = \sigma} 
			& \clx_{j+1} \ar@{<-}[r]^-{\sigma_1} 
			& \clx'_{j+1} \setminus \{\sigma\} 
				\cdots \ar@{<-}[r]^-{\sigma_{\ell - 2}} 
			& \clx'_{j+\ell-2} \setminus\{\sigma\} 
				\ar@{<-}[r]^-{\sigma_{\ell - 1}} 
			& \clx'_{j+\ell} \ar@{<-}[r]^-{\sigma_{\ell + 1}} 
			& \cdots\\
 		},
	\]
	where the removal of $\sigma = \sigma_\ell$ has been moved all the way up to 
	$\clx_i$. This can be attained by applying successively transposition 
	diamonds (Diagram~(\ref{eq:transposition_diamond})) in $\F_j[j;j+m]$, in 
	order to obtain $\F_{j+1}$. Studying the effect of transposition diamonds on 
	homology, algorithm~\cite{MariaO15} updates a homology matrix at $\clx_j$, 
	compatible with $\F_j$, into a homology matrix at $\clx_{j+1}$, compatible 
	with $\F_{j+1}$ defined above.
\end{enumerate}

\subsection{Adaptation to zigzag Morse filtrations}\label{ssec:morse_alg}

Using notations from Section~\ref{sec:hl_algo}, let $\oF$ be a general zigzag 
filtration:
\[
	\oF \ := \quad \xymatrix @C+8pt{
  		(\emptyset =) \,\oclx_1 \ar@{^{(}->}^-{\Sigma_1}[r] 
  		& \oclx_2 \ar@{<-^{)}}^-{\Sigma_2}[r] 
  		& \enspace \cdots \enspace \ar@{^{(}->}^-{\Sigma_{2k-1}}[r] 
  		& \oclx_{2k-1} \ar@{<-^{)}}^-{\Sigma_{2k}}[r] 
  		& \oclx_{2k} \,(= \emptyset)\\
  	}
\]
together with Morse matchings $\morse{i}$ on the set of cells $\Sigma_i$ of 
every \emph{forward} inclusion 
$\xymatrix{\oclx_{i} \ar@{->}[r]|-{\, \Sigma_i \, } & \oclx_{i+1}}$, $i$ odd. 

Let $\F$ be the associated \emph{atomic} zigzag filtration of complexes where 
all maps are forward or backward inclusions of a single cell: 
$
	\xymatrix @C-7pt{
		\F = \ \clx_{1} \ar@{<->}[r] 
		& \cdots 
		& \clx_{n} \ar@{<->}[l]
}$.

Algorithm~\cite{MariaO15} can update a homology matrix for a general complex 
using reflection and transposition diamonds to implement the insertion and 
deletion of cells pictured in Diagrams~(\ref{eq:simplecases}).
We now implement the operation of Diagram~(\ref{eq:hardcase}), introducing the 
chain map $\varphi_{\tau,\sigma}$.

At step $j$ of the algorithm, we maintain a zigzag Morse filtration $\M_j$ for 
the filtration $\F_j$. At space $\clx_j$, the filtration satisfies:

\begin{prop}[Zigzag Morse filtration $\M_j$]\label{prop:zzmorse}
	\ \\ 
	\vspace{-18px}
	\begin{enumerate}[itemsep=\smallskipamount]
  		\item The filtration $\M_j[1;j]$ is a general zigzag Morse filtration 
  		(defined in Section~\ref{sec:zzmorsefil}) for $\F[1;j]$ and its 
  		Morse matchings $\{\morse{i}\}_{i = 1 \ldots j}$,
  		\item the filtration $\M_j[j;j+m]$ is a standard Morse filtration 
  		(defined in~\cite{MischaikowN13} and Equation~(\ref{eq:stdmorse})) 
  		for the \emph{standard} filtration $\F_j[j;j+m]$.
	\end{enumerate}
\end{prop}

Before exhibiting the filtrations, we prove the following simple property of the 
zigzag persistence algorithm,

\begin{lemma}\label{lem:simpleprop}
	Let $\tau, \sigma$ be cells of $\clx_j$, and let 
	$\xymatrix @C-5pt{\clx_{p} \ar@{->}[r]^-{\tau} & \clx_{p+1}}$ and 
	$\xymatrix @C-5pt{\clx_{q} \ar@{->}[r]^-{\sigma} & \clx_{q+1}}$ be the two 
	maps in $\F$ that have the largest indices $1 \leq p,q < j$ for which a 
	forward inclusion of $\tau$ and $\sigma$, respectively, happens in $\F[1;j]$. 

	Let 
	$\xymatrix{\clx'_{j+r-1} \ar@{<-}[r]^-{\tau} & \clx'_{j+r}}$ and 
	$\xymatrix{\clx'_{j+s-1} \ar@{<-}[r]^-{\sigma} & \clx'_{j+s}}$, for indices 
	$1 \leq r,s \leq m$, be the backward inclusions of $\tau$ and $\sigma$ in 
	the part $\F_j[j;j+m]$ of the filtration $\F_j$. Then,
	\[
  		p < q \quad \text{ iff } \quad s < r.
	\]
	In other words, if $\tau$ is inserted before $\sigma$, it is removed after 
	$\sigma$.
\end{lemma}

\begin{proof}
	The only ``new'' arrows in the diagram are brought by the reflection 
	diamonds~(\ref{eq:reflection_diamond}) applied at index $j$ of the 
	algorithm, on $\F_j$, which induces the desired symmetry in forward and 
	backward arrows for the insertion of a given cell. We refer 
	to~\cite{MariaO15} for details on the algorithm.
\end{proof}

Now, consider the following diagram, where $(\tau,\sigma)$ are cells of $\clx_j$ 
which are paired in the Morse matching of $\clx_j$ induced by the Morse 
matchings $\{\morse{i}\}_{i = 1 \ldots j}$ of the filtration,
\begin{adjustbox}{minipage=1.2\textwidth,margin=0pt 8pt 0pt 0pt,center}
\begin{equation}\label{eq:bigdiag}
  	\begin{gathered}
  	\xymatrix @C-7pt @R-10pt{
		\F_j: \ 
			\ar@{<->}[r]|-{\cdots}  
    	& \clx_j  
    		\ar@{->}[r]^-{\id} 
    		\ar@{->}[d]^-{\psi_{\tau,\sigma} \circ \psi} 
    	& \clx_j 
    		\ar@{<-}[r]
    	 	\ar@{->}[d]^-{\psi} 
    	& \clx_{j+1}' 
    		\ar@/_20pt/[ll] 
    		\ar@{<-}[r]|-{\cdots} 
    		\ar@{->}[d]^-{\psi} 
    	& \clx'_{j+r} 
    		\ar@{<-}[r] 
    		\ar@{->}[d]^-{\psi} 
    	& \clx,\sigma,\tau 
    		\ar@{<-}[r]^-{\sigma} 
    		\ar@{->}[d]^-{\psi} 
    	& \clx,\tau 
    		\ar@{<-}[r]^-{\tau} 
    		\ar@{->}[d]^-{\psi} 
    	& \clx 
    		\ar@{<-}[r] 
    		\ar@/_20pt/[ll]|-{\ \{\sigma,\tau\} \ } 
    		\ar@{->}[d]^-{\psi} 
   	 	& \clx_{j+r-2} 
   	 		\ar@{<-}[r]|-{\cdots} 
   	 		\ar@{->}[d]^-{\psi} 
    	& \\
    	\bM_j: \ 
			\ar@{<->}[r]|-{\cdots}
   	 	& \A_j 
   	 		\ar@{->}[r]_-{\varphi_{\tau,\sigma}} 
   	 		\ar@{->}[d]^-{\id} 
    	& \A_j,\sigma,\tau 
    		\ar@{<-}[r] 
    		\ar@{->}[d]^-{\psi_{\tau,\sigma}}
    	& \A'_{j+1},\sigma,\tau 
    		\ar@{<-}[r]|-{\cdots} 
    		\ar@{->}[d]^-{\psi_{\tau,\sigma}}
    	& \A'_{j+r},\sigma,\tau 
    		\ar@{<-}[r] 
    		\ar@{->}[d]^-{\psi_{\tau,\sigma}}
    	& \A,\sigma,\tau 
    		\ar@{<-}[r]_-{\sigma} 
    		\ar@{->}[d]^-{\psi_{\tau,\sigma}}
    	& \A,\tau 
    		\ar@{<-}[r]_-{\tau}
    	& \A 
    		\ar@{<-}[r] 
    		\ar@{->}[d]^-{\id}
    	& \A'_{j+r-2} 
    		\ar@{<-}[r]|-{\cdots} 
    		\ar@{->}[d]^-{\id} 
    	& \\
		\M_j: \ 
			\ar@{<->}[r]|-{\cdots}
    	& \A_j 
    		\ar@{->}[r]
    	& \A_j 
    		\ar@{<-}[r]
    	& \A'_{j+1} 
    		\ar@{<-}[r]|-{\cdots} 
    	& \A'_{j+r} 
    		\ar@{<-}[r]
    	& \A 
    		\ar@{<-}[rr]
    	& \ 
    	& \A 
    		\ar@{<-}[r]
    	& \A'_{j+r-2} 
    		\ar@{<-}[r]|-{\cdots} 
    	& \\
  	}
  	\end{gathered}
\end{equation}
\end{adjustbox}
where arrows without label are simple inclusions of complexes. Simplifying 
notations, we denote by $\clx$ the complex $\clx'_{j+r-1}$, by $\A$ the complex 
$\A'_{j+r-1}$, and union of a complex and some cells by $\clx,\sigma,\tau$, 
instead of $\clx \cup \{\sigma,\tau\}$. We use this diagram until the end of the 
section, and define its various components progressively.

Lemma~\ref{lem:simpleprop} ensures that $\tau$ and $\sigma$, that are 
consecutively inserted (Morse pair, Diagram~(\ref{eq:type2})), are consecutively 
removed in $\F_j[j;j+m]$, as pictured above. The filtration $\F_j$ appears on 
top, where two arrows (curved horizontal) are further decomposed for 
convenience. 

By induction, let $\M_j$ be the zigzag Morse filtration maintained by the 
algorithm at step $j$, and satisfying Properties~\ref{prop:zzmorse}. Performing 
reflection diamonds~(\ref{eq:reflection_diamond}) at index $j$, and 
transposition diamonds~(\ref{eq:transposition_diamond}) at indices 
$j + r$, $r > 0$, maintains the Properties~\ref{prop:zzmorse}. Consequently, at 
the level of the zigzag Morse filtration, the zigzag algorithm~\cite{MariaO15} 
can implement insertions and deletions of critical cells 
(Diagrams~(\ref{eq:simplecases})) with no further modification, while 
maintaining a Morse filtration $\M_j \mapsto \M_{j+1}$ satisfying the 
algorithmic invariant Properties~\ref{prop:zzmorse}.

The only obstruction to using the zigzag persistence algorithm is the operation 
introduced in Diagram~(\ref{eq:hardcase}). Consequently, consider the next 
operation in $\F$ to be the removal 
$\xymatrix{\clx_{j} \ar@{<-}[r]|-{\, \sigma \, } & \clx_{j+1}}$ of a 
non-critical cell $\sigma$, paired with a cell $\tau$ in the Morse matching of 
$\clx_j$, such that $\tau$ is not removed.
The cell $\sigma$ cannot be ``directly removed'' as it does not appear in 
$\M_j[j;j+m]$. We focus the rest of this section to the definition and study of 
the zigzag Morse filtration $\bM_j$ of Diagram~(\ref{eq:bigdiag}). 

Let $\bM_j$ be as above, where the map $\varphi_{\tau,\sigma}$ is the map 
defined in Diagram~(\ref{eq:hardcase}), and the chain maps $\psi$ between $\F_j$ 
and $\bM_j$ are the ones of Diagrams~(\ref{eq:simplecases}) 
and~(\ref{eq:hardcase}). By Theorem~\ref{thm:equivalence}, these maps induce 
an isomorphism of zigzag modules $\Hom(\F_j) \to \Hom(\bM_j)$, and the 
filtrations have same persistent homology. Additionally, $\bM_j$ is a zigzag 
Morse filtration, and a standard Morse filtration from space $\A_j,\sigma,\tau$ 
on to the right, i.e., it satisfies Properties~\ref{prop:zzmorse}. Finally, 
$\sigma$ is critical in $\A_j,\sigma,\tau$, and can be removed with the zigzag 
persistence algorithm to obtain $\M_{j+1}$. 

\paragraph{Compatible homology matrix.} 
We design an algorithm to turn a homology matrix at $\A_j$, compatible with 
$\M_j$, into a homology matrix at $\A_j \cup \{\tau,\sigma\}$, compatible with 
$\bM_j$, in Diagram~(\ref{eq:bigdiag}).

Consider $\clx_j$ in $\F_j$ (Diagram~(\ref{eq:bigdiag})), containing $m$ cells:

\begin{definition}[\cite{DBLP:journals/corr/abs-1107-5665}]\label{def:hommat}
	Let $\clx$ be a cell complex of size $m$ and $\mtx = \{c_0, \ldots , c_{m-1} \}$ 
	be a collection of $m$ chains of $C(\clx)$. We say that $\mtx$ is a 
	\emph{homology matrix} at $\clx$ if there exists an ordering 
	$\sigma_0, \ldots, \sigma_{m-1}$ of the $m$ cells of $\clx$ such that:
	\begin{enumerate}[start=0]
		\item for all $0 \leq r < m$, the restriction 
		$\{\sigma_0, \ldots, \sigma_r\} \subset \clx$ is a subcomplex of $\clx$,
		\item for all $0 \leq r < m$, the leading term of $c_r$ is $\sigma_r$ for 
		the chosen ordering, i.e., 
		$c_r = \varepsilon_0 \sigma_0 + \ldots + \varepsilon_{r-1} \sigma_{r-1} 
		+ \sigma_r$, for some $\varepsilon_i \in \field$,
	\end{enumerate}
	and there exists a partition $\{0, \ldots, m-1\} = F \sqcup G \sqcup H$, and 
	a bijective pairing $G \leftrightarrow H$, satisfying:
	\begin{enumerate}[start=2]
		\item for all indices $f \in F$, $\bo^{\clx_j} c_f = 0$,
		\item for all pairs $g \leftrightarrow h$ of $G \times H$, 
		$\bo^{\clx_j} c_h = c_g$.
	\end{enumerate}
\end{definition}

This data encodes~\cite{DBLP:journals/corr/abs-1107-5665} the persistent 
homology of the (standard) filtration $\F_j[j;j+m]$. In particular, the homology 
groups of $\clx_j$ are equal to $\langle [c_f] : f \in F \rangle$. It is 
convenient to see this data as a matrix $M_{\mtx}$ with cycle $c_i$ as $i^{th}$ 
column, expressed in the basis $\{\sigma_i\}_{i = 1 \ldots m}$ for rows. In this 
case, condition~(1) of the definition is equivalent to the matrix being upper 
triangular, with no zero entry in the diagonal.

Additionally,

\begin{definition}[\cite{MariaO15}]\label{def:compatible}
	We denote by $\bigoplus_\ell \Imod[b_\ell;d_\ell]$ the interval decomposition 
	of $\Hom(\F_j)$. A homology matrix $\mtx = \{c_0, \ldots , c_{m-1} \}$ at 
	$\clx_j$ is {\em compatible} with the filtration $\F_j$ iff there exists a 
	zigzag module isomorphism 
	$\Phi \colon \Hom(\F) \to \bigoplus_\ell \Imod[b_\ell;d_\ell]$ such that 
	$\Phi_j \colon \Hom(\clx_j) \to \field^{|F|}$ sends $\{[c_f] : f \in F\}$ to 
	the canonical basis of $\field \times \cdots \times \field$.
\end{definition}

The Morse theory algorithm for persistent homology of~\cite{MischaikowN13} can 
be applied to maintain a compatible homology matrix for a Morse filtration under 
the operations pictured in Diagrams~(\ref{eq:simplecases}). We design the update 
for the new operation of Diagram~(\ref{eq:hardcase}). Consider:
\[
	\xymatrix @C-5pt{
		\M_j: \ \A_{1} \ar@{<->}[r] & \cdots & \A_{j} \ar@{<->}[l]\\
	} 
	\quad \text{and} \quad
	\xymatrix @C-5pt{
		\F_j: \ \clx_{1} \ar@{<->}[r] & \cdots & \clx_{j} \ar@{<->}[l]\\
	},
\]
such that $\M_j$ is a zigzag Morse filtration for $\F_j$. Assume $\A_j$ has $m$ 
cells, and let $\mtx = \{c_0, \ldots , c_{m-1}\}$ be a homology matrix at $\A_j$ 
compatible with $H(\M_j)$. Following Diagram~(\ref{eq:hardcase}), consider: 
\[
	\xymatrix @C-10pt{
		\bM_j: \,\A_{1} \ar@{<->}[r] 
		& \cdots 
		& \A_{j} \ar@{<->}[l] \ar@{->}[r]
		& \A_j \cup \{\tau,\sigma\}
	} 
	\enspace \text{and} \enspace
	\xymatrix @C-10pt{
		\oF_j: \,\clx_{1} \ar@{<->}[r] 
		& \cdots 
		& \clx_{j} \ar@{<->}[l] \ar@{->}[r]^-{\id} 
		& \clx_j
	}
\]
such that $\bM_j$ is a zigzag Morse filtration for $\oF_j$. From $\mtx$, we 
define a homology matrix $\bmtx := \{c'_0, \ldots , c'_{m-1},c_\tau, c_\sigma\}$ 
at $\A_j \cup \{\tau,\sigma\}$ that is compatible with $H(\bM_j)$.

Denote the two last complexes and their boundary maps in $\bM_j$ by $(\A_j, \bo)$ 
and $(\A_j' , \bo')$, with $\A_j' := \A_j \cup \{\tau,\sigma\}$. Then:
\begin{itemize}
	\item for all indices $i \in F \sqcup H$, define 
	\[
		c'_i := c_i - \left(\inc{\sigma}{\tau}^{\clx_j}\right)^{-1} 
		\left(\sum_{\nu \in c_i} \inc{\nu}{\tau}^{\A'} \right) \cdot \sigma,
	\]
	where the sum is taken over all cells $\nu$ in the support of chain $c_i$,
	
	\item define $c_\tau := \bo' \sigma$, and $c_\sigma := \sigma$, and put the 
	index of $c_\tau$ in $G$, the index of $c_\sigma$ in $H$, and pair them 
	together,
	
	\item the pairing $G \leftrightarrow H$ inherited from $\mtx$ remains 
	unchanged, and so does $F$.
\end{itemize}

\begin{lemma}\label{lem:updategiveshommat}
	The collection $\bmtx$ is a homology matrix at $\A_j \cup \{\tau,\sigma\}$ in 
	Diagram~(\ref{eq:bigdiag}).
\end{lemma}

\begin{proof}
	We prove that $\bmtx$ satisfies the conditions of Definition~\ref{def:hommat}.

	\begin{enumerate}[itemsep=\medskipamount, start=0]
		\item Because a Morse matching induces an acyclic Hasse Diagram, there exists 
		$r$ such that 
		$\sigma_0,\ldots,\sigma_r,\tau,\sigma,\sigma_{r+1},\ldots,\sigma_{m-1}$ is an 
		ordering of the cells of $\A_j \cup \{\tau,\sigma\}$ such that the first 
		$k$ cells form a subcomplex, for any $k$, as in Definition~\ref{def:hommat}.
		
		\item 
		\textbf{Case $c_\tau, c_\sigma$.} 
		The leading term of $c_\sigma$ is $\sigma$. We prove that the leading 
		term of $c_\tau$ is $\tau$ in the ordering defined above. 
		Let $\hasse$ be the oriented Hasse diagram of $\clx_j$ for the Morse 
		matching where $(\tau,\sigma)$ forms a Morse pair (complex $\A_j$), and 
		$\hasse'$ for the matching where $\tau$ and $\sigma$ are critical 
		(complex $\A_j \cup \{\tau,\sigma\}$); they differ by the orientation of 
		arrow $\sigma \leftrightarrow \tau$. 
		First, $\langle \bo' \sigma, \tau \rangle^{\A_j \cup \{\tau,\sigma\}} \neq 0$ 
		because there exists a unique gradient path from critical cell $\sigma$ 
		to critical cell $\tau$ in $\A_j \cup \{\tau,\sigma\}$, which is the one edge 
		path $\gamma = (\tau,\sigma)$. The path $\gamma$ exists because $\tau$ is a 
		facet of $\sigma$ in $\clx_j$. If there were another distinct gradient 
		path from $\sigma$ to $\tau$ in $\hasse'$, not containing the edge 
		$\sigma \rightarrow \tau$, this path would exist in $\hasse$ and form a 
		cycle with edge $\tau \rightarrow \sigma$ in $\hasse$; a contradiction 
		with the definition of Morse matchings.
		Second, if $\mu \in \A_j \cup \{\tau,\sigma\}$, is critical such that 
		$\inc{\sigma}{\mu}^{\A_j \cup \{\tau,\sigma\}} \neq 0$, then $\mu$ appears 
		before $\sigma$ (and $\tau$) in the ordering. Indeed, there exists a gradient 
		path 
		$\gamma = (\sigma, \mu_1, \w(\mu_1), \ldots, \w(\mu_{r-1}), \mu_r = \mu)$ 
		from $\sigma$ to $\mu$ in $\hasse'$. The cells $(\mu_i,\w(\mu_i))$ of a 
		pair are inserted consecutively by construction, and, for all $i$, 
		$\mu_i$ is inserted before $\w(\mu_{i-1})$ because it is a facet in 
		$\clx_j$. By transitivity, $\mu$ is inserted before $\sigma$.

		\textbf{Case $c'_i$.}
		The leading term of $c'_i$ is $\sigma_i$. If $c'_i = c_i$, it is direct. 
		Otherwise, by construction, 
		$c'_i = c_i + \alpha \cdot \sigma$, $\alpha \neq 0$, and the chain $c_i$ 
		contains cells $\nu$ in its support such that 
		$\inc{\nu}{\tau}^{\A_j \cup \{\tau,\sigma\}} \neq 0$, i.e., cofacets of 
		$\tau$ in $\A_j \cup \{\tau,\sigma\}$. With a similar transitivity argument 
		as above, $\tau$ 
		(and $\sigma$) must consequently appear before such $\nu$ in the ordering 
		of cells defined. The leading term of $c'_i$ is then unchanged.
		
		\item Let $c_i$ be a chain such that $i \in F \sqcup H$. 
		By Lemma~\ref{lem:nbo}, it is a direct calculation from the definition of 
		$c'_i$ that $\bo' c'_i = \bo c_i$. Consequently, Conditions~(2) and~(3) 
		of Definition~\ref{def:hommat} are satisfied for those chains. The 
		pairing $G \leftrightarrow H$ remains valid, because 
		$\bo' c'_h = \bo c_h = c_g = c'_g$ for 
		$g \leftrightarrow h$, $(g,h) \in G \times H$.
		
		\item By definition, $\bo' c_{\sigma} = c_\tau$, 
		their indices are in $H \times G$ and paired together.		
		\qedhere
	\end{enumerate}
\end{proof}

We now prove the compatibility condition:

\begin{lemma}\label{lem:comp}
	The homology matrix $\bmtx$ at $\A_j \cup \{\tau,\sigma\}$ is compatible with 
	$\bM_j$ in Diagram~(\ref{eq:bigdiag}).
\end{lemma}

\begin{proof}
	By hypothesis, $\mtx = \{ c_0, \ldots, c_{m-1}\}$ is a homology matrix at 
	$A_j$, compatible with $\M_j$; let 
	$\Omega \colon \Hom(\M_j) \to \oplus_\ell \Imod[b_\ell;d_\ell]$ be a zigzag 
	module isomorphism such that $\Omega_j$ sends $\{[c_f] : f \in F\}$ to the 
	canonical basis of $\field \times \ldots \times \field$. 

	Note that, none of the $c'_i$ have an entry $\tau$, except for $c_\tau$, 
	whose index is in $G$ by construction. Consequently, by 
	Properties~\ref{prop:phi}, the chain map 
	$\psi_{\tau,\sigma} \colon \chain(\A_j,\sigma,\tau) \to \chain(\A_j)$ simply 
	cancels the entry $\sigma$ in every $c'_f$, $f \in F$, and 
	$\psi_{\tau,\sigma} c'_f = c_f$. Consequently, consider the chain maps 
	between $\bM_j$ and $\M_j$ in Diagram~(\ref{eq:bigdiag}). Each square 
	commutes by virtue of Theorem~\ref{thm:filteredmorse} (for inclusions) and 
	Lemma~\ref{lem:rm_half_pair} (for $\varphi_{\tau,\sigma}$), and they induce 
	an isomorphism $\Phi_* \colon \Hom(\bM) \to \Hom(\M)$ of zigzag modules. The 
	isomorphism 
	$\Omega \circ \Phi_* \colon \Hom(\bM) \to \oplus_\ell \Imod[b_\ell;d_\ell]$ 
	sends $\{[c'_f] : f \in F\}$ to the canonical basis of 
	$\field \times \ldots \times \field$, and $\bmtx$ is compatible with $\bM$.
\end{proof}

\begin{algorithm}[t!]
	\SetKwFunction{ZPA}{zigzag\_persistence\_algorithm}
  	\SetKwInOut{Output}{output}
  	\SetKwInOut{Input}{input}
  	\Input{atomic zigzag filtration $\F \ := \ \ \ \xymatrix @C+8pt{
  		(\emptyset =) \ \clx_1 \ar@{<->}[r] 
  		& \clx_2 \ar@{<->}[r] 
  		& \cdots \ar@{<->}[r] 
  		& \clx_{n-1} \ar@{<->}[r] 
  		& \clx_n \ (= \emptyset) \\
  	}$}
  	\Output{persistence diagram of $\F$}
  	set $M_{\mtx} \leftarrow \emptyset$\;
  	\For{$j=1 \ldots n-1$}{
    	\If{
    		$\xymatrix@C+5pt{
    			\clx_j \ar@{<->}[r]^-{\sigma} & \clx_{j+1}
    		}$, $\sigma \in \clx_j$ critical
    	}{
      		use \ZPA{$M_{\mtx}$, $\M_j$, $\sigma$} to add or remove $\sigma$\;
    	}
    	\If{
    		$\xymatrix@C+5pt{
    			\clx_j \ar@{<->}[r]^-{\{\tau,\sigma\}} & \clx_{j+1}
    		}$, $(\tau,\sigma)$ Morse pair
    	}{
      		do nothing\;
    	}
    	\If{
    		$\xymatrix@C+5pt{
    			\clx_j \ar@{->}[r]^-{\id} 
    			& \clx_j \ar@{<-}[r]^-{\sigma} 
    			& \clx_{j+1}
    		}$, $\sigma$ paired with $\tau$, $\tau$ not removed
    	}{
    		set $M_{\mtx} \leftarrow M_{\bmtx}$ as described above\;
    		use \ZPA{$M_{\mtx}$,$\bM_j$,$\sigma$} to remove $\sigma$\;
    	}
  	}
  	\caption{Zigzag persistence algorithm for Morse filtrations}
  	\label{alg:morse_partition}
\end{algorithm}

In conclusion, for an input atomic zigzag operation $\F$, with three atomic maps 
pictured in Diagrams~(\ref{eq:type1}),~(\ref{eq:type2}), and~(\ref{eq:type3}), 
the Morse algorithm for computing the zigzag persistence of $\F$ is given in 
Algorithm~\ref{alg:morse_partition}, where the routine 
\ZPA{$M_{\mtx}$, $\M_j$, $\sigma$} is the zigzag persistence algorithm 
of~\cite{MariaO15} to handle forward or backward insertions of a single cell in 
a homology matrix $M_{\mtx}$ at complex $\A_j$, compatible with the filtration 
$\M_j$ (see Diagram~\ref{eq:bigdiag}). Each iteration of the {\tt for} loop 
turns a homology matrix $M_{\mtx}$ at complex $\A_j$, compatible with the 
filtration $\M_j$, into a homology matrix at complex $\A_{j+1}$, compatible with 
the filtration $\M_{j+1}$, where $\M_{j+1}$ is a zigzag Morse filtration for 
$\F_{j+1}$, and $\A_{j}$ and $\A_{j+1}$ are respectively Morse complexes for 
$\clx_j$ and $\clx_{j+1}$.

\paragraph{Implementation and complexity.} 
We represent $\mtx = \{c_0, \ldots , c_{m-1}\}$ by an $(m\times m)$-sparse 
matrix data structure $M_{\mtx}$. Assume computing boundaries and coboundaries 
in a Morse complex of size $m$ is given by an oracle of complexity $\orac(m)$. 
We implement the transformation 
$\mtx = \{c_0, \ldots, c_{m-1}\} \to \bmtx 
= \{c'_0, \ldots,\allowbreak c'_{m-1}, c_\tau, c_\sigma\}$ 
presented above by:
\begin{itemize}
	\item computing the boundary $\bo' \sigma$ of $\sigma$ in 
	$\A_j \cup \{\tau,\sigma\}$, and the coboundary 
	$\{\nu : \inc{\nu}{\tau}^{\A_j \cup \{\tau,\sigma\}} \neq 0 \}$ of $\tau$, 
	in $O(\orac(m))$ operations,
	\item adding columns $c_\tau$ and $c_\sigma$ to the matrix in $O(m)$ 
	operations,
	\item computing $c'_i$ for all $i$, in $O(m^2)$. We can restrict the 
	transformation to those $c_i$ containing a cell of the coboundary of $\tau$.
\end{itemize}

Consequently, we can perform the transformation above in $O(m^2 + \orac(m))$ 
operations on a $(m \times m)$-matrix. The zigzag persistence algorithm 
of~\cite{CarlssonSM09,MariaO15} deals with forward and backward insertions of a 
single cell in $O(m^2)$ operations. 

In conclusion, let 
$\oF = (
	\xymatrix@C+5pt{
		\oclx_{i} \ar@{<->}[r]|-{\, \Sigma_i \, } 
		& \oclx_{i+1}
})_{i=1 \ldots 2k}$ 
be a general zigzag filtration (Diagram~(\ref{eq:genzzfil})), and let $\M$ be a 
zigzag Morse filtration as defined in Section~\ref{sec:hl_algo}, for a 
collection of Morse matchings $\morse{i}$ on $\Sigma_i$, $i$ odd. And:
\begin{itemize}
	\item denote by $n$ the total number of insertions and deletions critical 
	cells in $\M$, and by $|\A_m|$ the maximal number of critical cells of a 
	complex in $\M$,
	\item denote by $N$ the total number of insertion and deletion of cells in 
	$\oF$, and by $|\clx_m|$ the maximal number of cells of a complex in $\oF$.
\end{itemize}

Additionally, we compute Morse matchings using the fast coreduction algorithm of 
Mrozek and Batko~\cite{DBLP:journals/dcg/MrozekB09}. Even if computing optimal 
Morse matchings is hard in general~\cite{DBLP:journals/siamdm/JoswigP06}, this 
heuristic gives experimentally very small Morse complexes, with constant 
amortized cost per cell considered. We compute boundaries and coboundaries in a 
Morse complex $\A$ of a complex $\clx$ by a linear traversal of the Hasse 
diagram of $\clx$. We store in memory the homology matrix of the Morse complex 
and the complex $\clx$. Consequently, the total cost of the algorithm is:

\begin{theorem}
 	The persistent homology of $\oF$ can be computed in 
 	
 	\begin{itemize*}[itemjoin=\hspace{20px}]
		\item time: $O(n \cdot |\A_m|^2 + n \cdot |\clx_m| + N)$,
		\item memory: $O(|\A_m|^2 + |\clx_m|)$.
	\end{itemize*}
\end{theorem}

In comparison, running the (practical) zigzag persistence 
algorithms~\cite{CarlssonS10,CarlssonSM09,MariaO15} require 
$O(N \cdot |\clx_m|^2)$ operation and memory $O(|\clx_m|^2)$.

%% file: experiments.tex
\begin{table}[t]
	\centering
	\small
	\textsf{
	\makebox[\textwidth][c]{
	\begin{tabular}{|l|*{8}{r|}}
\cline{2-9}
\multicolumn{1}{l|}{}
	& \multicolumn{4}{c|}{\texttt{Without Morse reduction}}
	& \multicolumn{4}{c|}{\texttt{With Morse reduction}}\\
\cline{2-9}
\multicolumn{1}{l|}{}
	& \makecell*[c]{$N$ \\ $\times 10^6$}	& \makecell*[c]{$|\clx_m|$}
	& \makecell*[c]{time (s) \\ cpx + pers}	
	& \makecell*[c]{mem. \\ peak \\ (GB)}
	& \makecell*[c]{$n$ \\ $\times 10^6$}	& \makecell*[c]{$|\A_m|$}
	& \makecell*[c]{time (s) \\ cpx + pers}	
	& \makecell*[c]{mem. \\ peak \\ (GB)} \\
\hline
\texttt{KlBt5}
	& 63.3							& 187096
	& 403 + 2912
	& 4.7
	& 4.9							& 11272
	& 394 + 448
	& 1.1	\\
\texttt{Spi3}
	& 66.1							& 47296
	& 435 + 4438
	& 5.2
	& 3.8							& 12810
	& 382 + 343
	& 1.1	\\
\texttt{MoCh}
	& 75.7 							& 37709
	& 460 + 4680
	& 5.8
	& 4.1 							& 11975
	& 450 + 318
	& 1.1	\\
\texttt{Sph3}
	& 99.4 							& 66848
	& 430 + 3498
	& 7.5
	& 4.2 							& 13432
	& 665 + 853
	& 1.3	\\
\texttt{To3}
	& 32.8							& 32903
	& 117 + 847
	& 2.4
	& 1.6 							& 7570
	& 173 + 79
	& 0.47	\\
\texttt{By}
	& 30.5							& 18764
	& 153 + 951
	& 2.3
	& 5.2							& 8677
	& 165 + 287
	& 0.96	\\
\hline
	\end{tabular}
	}
	}
	\caption{Experimental results for the oscillating Rips zigzag filtrations. 
	For each experiment, the maximal dimension is $10$, $\mu = 4$, $\nu = 6$, 
	except for \texttt{Sph3}, where $\nu = 7$. The number of vertices is $2000$.
	}\label{tab:results_osc}
\end{table}

In this section, we report on the performance of the zigzag persistence 
algorithm~\cite{MariaO15} with and without Morse reduction. The corresponding code 
will be avaible in a future release of the open source library 
\texttt{GUDHI}~\cite{gudhi:urm}.

The following tests are made on a 64-bit Linux (Ubuntu) HP machine
with a 3.50 GHz Intel processor and 63 GB RAM. The programs are all 
implemented in C++ and compiled with optimization level \texttt{-O2} and 
\texttt{gcc-8}. Memory peaks are obtained via the \texttt{/usr/bin/time -f} 
Linux command, and timings are measured via the C++ 
\texttt{std::chrono::sys\-tem\_clock::now()} method.
The timings for File IO are not included in any process time.

We run two types of experiments: homology inference from point clouds, using 
oscillating Rips zigzag filtrations, and levelset persistence of 3D-images. Both 
applications are described in the introduction.

For homology inference, we use both synthetic and real data points. 
The point clouds \texttt{KlBt5}, \texttt{Spi3}, \texttt{Sph3}, and \texttt{To3} 
are synthetic samples of respectively the 5-dimensional Klein bottle, a 
3-dimensional spiral wrapped around a torus, the 3-dimensional sphere, and the 
3-dimensional torus. The point cloud \texttt{MoCh} and \texttt{By} are 
3-dimensional measured samples of surface models: the MotherChild model, and the 
Stanford bunny model from the Stanford Computer Graphics Laboratory. The results 
with corresponding parameters are presented in Table~\ref{tab:results_osc}.

Levelset persistence is computed for a function $f: [0;1]^3 \to \R$, were $f$ is 
a Fourier sum with random coefficients, as proposed in the \texttt{DIPHA} 
library\footnote{
\url{github.com/DIPHA/dipha/blob/master/matlab/create_smooth_image_data.m}} 
as representative of smooth data. The cube $[0;1]^3$ and function $f$ are 
discretized into equal size voxels. For some tests, we also added 
random noise to the values of $f$. The values of $s_1 \leq s_2 \leq \ldots$ are 
spaced out equally such that $s_{i + 1} - s_i = \epsilon$ for all $i$.
The results with corresponding parameters are presented in 
Table~\ref{tab:results_ls}.

In all experiments, timings are decomposed into `cpx' for computation de-dicated 
to the complex (construction, computation of (co)boundaries and of Morse 
matchings) and `pers' for the computation of zigzag persistence.

\begin{table}[t]
	\centering
	\small
	\textsf{
	\makebox[\textwidth][c]{
	\begin{tabular}{|r|*{9}{r|}}
\cline{3-10}
\multicolumn{2}{l|}{}
	& \multicolumn{4}{c|}{\texttt{Without Morse reduction}}
	& \multicolumn{4}{c|}{\texttt{With Morse reduction}}\\
\cline{1-10}
\makecell*[c]{$\epsilon$}		& \makecell*[c]{max. \\ noise}	
	& \makecell*[c]{$N$ \\ $\times 10^6$}	& \makecell*[c]{$|\clx_m|$}
	& \makecell*[c]{time (s) \\ cpx + pers}	
	& \makecell*[c]{mem. \\ peak \\ (GB)}
	& \makecell*[c]{$n$ \\ $\times 10^6$}	& \makecell*[c]{$|\A_m|$}
	& \makecell*[c]{time (s) \\ cpx + pers}	
	& \makecell*[c]{mem. \\ peak \\ (GB)} \\
\hline
0.1								& 0
	& 34							& 286780
	& 563 + 1725
	& 3.9
	& 6.3							& 48578
	& 224 + 29
	& 2.7	\\
0.15							& 0
	& -									& -
	& $\infty$
	& -
	& 9.3							& 115558
	& 756 + 44
	& 3.6	\\
0.15							& 0.5
	& 36.5							& 315305
	& 417 + 3248
	& 4.2
	& 4.7							& 36144
	& 221 + 59
	& 2.8	\\
0.2								& 0
	& -									& -
	& $\infty$
	& -
	& 15.5							& 245360
	& 2097 + 68
	& 4.7	\\
0.2								& 0.5
	& -									& -
	& $\infty$
	& -
	& 5.6							& 56500
	& 392 + 47
	& 3.4	\\
\hline
	\end{tabular}
	}
	}
	\caption{Experimental results for the level set zigzag filtrations. 
	For each experiment, the function $f: [0;1]^3 \to [-14,21]$ is applied to 
	$129^3 = 2\,146\,689$ cells and the persistence is computed for maximal 
	dimension~$3$. The interval size is denoted by $\epsilon$. 
	The infinity symbol $\infty$ corresponds to more than 12 hours computing 
	time.}
	\label{tab:results_ls}
\end{table}

\paragraph{Analysis of the results.} 
The results show a significant improvement when using Morse reduction. 
For homology inference (Table~\ref{tab:results_osc}), the total running time is 
between 2.5 and 6.7 times faster when using Morse reduction. Moreover, most of 
the computation is transferred onto the computation of the Morse complex, which 
opens new roads to improvement in future implementation, such as parallelization 
of the Morse reduction~\cite{8440824} (note that parallelization of the 
computation of zigzag persistence is not possible in the streaming model). 
In particular, the computation of zigzag persistence is from 3.3 to 14.7 times 
faster. The better performance is due to filtrations being from 5.8 to 23.5 
times shorter than the original ones (quantities $n$ vs $N$ in the complexity 
analysis) and smaller complexes, from 2.2 to 16.6 times smaller with the Morse 
reduction (quantities $|\A_m|$ and $|\clx_m|$ in the complexity analysis). Note 
that the memory consumption with Morse reduction is from 2.4 and up to 5.6 times 
smaller, which is critical on complex examples in practice.

For levelset persistence (Table~\ref{tab:results_ls}), the total running time is 
at least 9 times faster, and the computation of zigzag persistence alone is 
itself approximatively 55 times faster, when the computation without Morse 
reduction finished. On those cases that finish, the filtration size is from 5.5 
to 7.7 times shorter with Morse reduction, the maximal size of the complexes 
between 5.9 and 8.7 times smaller, and the memory consumption around $50\%$ more 
efficient.

Additionally, using Morse reduction allows to handle cases where the standard 
zigzag algorithm never finishes (more than 12 hrs). On these examples, the Morse 
algorithm does not take more than 36 min. for the entire computation.

These results agree with the complexity analysis (Section~\ref{sec:ll_algo}) 
where terms $O(|\A_m|^2)$ and $O(|\clx_m|^2)$ dominate both time and memory 
complexities.